\documentclass[sigconf]{acmart}

\renewcommand\footnotetextcopyrightpermission[1]{} 
\usepackage{fixltx2e}
\usepackage{dblfloatfix}
\usepackage{graphicx}
\graphicspath{{./Figs/}} 
\usepackage{amsmath,amsfonts,amssymb,amscd,amsthm,xspace}

\DeclareMathOperator*{\argmin}{arg\,min}
\usepackage{mathrsfs}
\usepackage{mathtools} 
\usepackage{dsfont}
\usepackage{booktabs}       
\usepackage{nicefrac}       
\usepackage[]{algorithm2e}
\usepackage{ifthen}
\newboolean{showcomments}\setboolean{showcomments}{false}

\newtheorem{theorem}{\textit{Theorem}}
\newtheorem{lemma}[theorem]{\textit{Lemma}}
\newtheorem{corollary}[theorem]{\textit{Corollary}}

\newtheorem{definition}{\textit{Definition}}
\numberwithin{definition}{section}

\theoremstyle{definition}
\newtheorem{example}{{Example}}
\numberwithin{example}{section}

\newtheorem{remark}{{Remark}}
\newcommand{\nn}{n}
\newcommand{\ns}{W}
\newcommand{\nt}{T}

\begin{document}
\title{Real-time Flexibility Feedback for Closed-loop\\ Aggregator and System Operator Coordination}

\author{{Tongxin Li}}
\affiliation{%
  \institution{CMS, California Institute of Technology}}
\email{tongxin@caltech.edu}
\author{{Steven H. Low}}
\affiliation{%
  \institution{CMS, EE, California Institute of Technology}}
\email{slow@caltech.edu}
\author{{Adam Wierman}}
\affiliation{%
  \institution{CMS, California Institute of Technology}}
\email{adamw@caltech.edu}

\begin{abstract}
Consider a system operator that wishes to optimize its objectives over
time subject to 
operational constraints as well as private constraints of controllable loads
managed by an aggregator.  In this paper, we design a 
\textit{real-time} feedback signal for the aggregator
to quantify and communicate its available flexibility 
to the system operator. 
The proposed feedback signal at each time is the conditional 
probability of future feasible trajectories that will be enabled
by the operator's decision.  We show that it is the unique distribution
that maximizes a system capacity for flexibility.  It allows the
system operator to maintain feasibility and enhance future flexibility while
optimizing its objectives.
We illustrate how the design can be used by the system operator to perform online cost minimization and real-time capacity estimation, while provably satisfying the private constraints of the loads.   

\end{abstract}



\maketitle

\section{Introduction}

The need to manage the uncertainty and volatility caused by the growing penetration of renewable sources such as wind and solar power has created a desire to increase the ability of the system to provide flexibility via distributed energy resources (DERs) and aggregators have emerged as dominate players for coordinating these loads~\cite{callaway2010achieving,burger2017review}. 
The power of aggregators is that they are able to provide coordination among large pools of DERs and then give a single point of contact for independent system operators (ISOs) to call on for flexibility.  This enables ISOs to minimize cost, respond to unexpected fluctuations of renewables, and even mitigate failures quickly and reliably.  

To realize the potential benefits of aggregators, ISOs need to be able to call on the aggregator via a time-varying signal, e.g., a desired power profile, that satisfies the operational constraints and optimizes a system objective.  The signal is then disaggregated by the aggegator in order to determine the behavior of the loads under its control.  However, the loads have private constraints on their operation (e.g., 
satisfying energy demands of electric vehicles before their deadlines). These constraints limit the flexibility available to the aggregator and so the aggregator must also communicate with the ISO by providing a signal that
quantifies its available flexibility. This signal is of crucial importance for the ISO when determining the signal it sends to the aggregator, and thus the aggregator and the ISO form a closed-loop control system. 

This paper focuses on the design of this closed-loop system and, in particular, the design of the signal quantifying the available flexibility sent from the aggregator to the ISO.  The question of how to design the signal providing information on aggregate flexibility of the aggregator to the operator, namely \emph{the flexibility feedback signal}, is complex and has been the subject of significant research over the last decade, e.g.,~\cite{hao2014characterizing,hao2014aggregate,sajjad2016definitions,zhao2017geometric,madjidian2018energy,chen2018aggregating,sadeghianpourhamami2018quantitive,evans2019graphical,bernstein2016aggregation}.  Any feedback design must balance between a variety of conflicting goals.  In particular a design must be:
\begin{enumerate}
    \item \textbf{Concise.} Given the scale of aggregators and the complexity of the constraints of loads, it is impossible to communicate precise information about every load.  Instead, aggregate flexibility feedback must be a concise summary of a system's constraints. 
    Even if it was possible, providing exact information about the constraints of each load governed by the aggregator would not be desirable because the load constraints are typically private.  Information conveyed to the ISO must limit the leakage about specific load constraints. 
    \item \textbf{Informative.} The feedback sent by an aggregator needs to be informative enough that it allows the ISO to  achieve operational objectives, e.g., minimize cost, and, most importantly, guarantee the feasibility of the whole system with respect to the private load constraints. 
    \item \textbf{General.} Any design for a flexibility feedback signal must be general enough to be applicable for a wide variety of controllable loads, e.g., electric vehicles (EVs), heating, ventilation, and air conditioning (HVAC) systems,  energy storage units, thermostatically controlled loads, residential loads, and pool pumps.  It is impractical to imagine different feedback signals for each load, so the same design must work for all DERs.
\end{enumerate}
The challenge and importance of the design of flexibility feedback signals has led to the emergence of a rich literature. In many cases, the literature focuses on specific classes of controllable loads, such as electric vehicles (EVs)~\cite{wenzel2017real}, heating, ventilation, and air conditioning (HVAC) systems~\cite{hao2014ancillary},  energy storage units~\cite{evans2019graphical}, thermostatically controlled loads~\cite{hao2014aggregate} or residential loads and pool pumps~\cite{sajjad2016definitions,meyn2015ancillary}.  In the context of these applications, there have been a variety of approaches suggested, e.g., convex geometric approximations~\cite{hao2014aggregate,zhao2017geometric,evans2019graphical,chen2018aggregating,chen2018aggregating},
scheduling based aggregation~\cite{subramanian2012real,subramanian2013real,papadaskalopoulos2013decentralized}, and probability-based characterization~\cite{sajjad2016definitions,meyn2015ancillary}.  These approaches have all yielded some success, especially in terms of quantifying the aggregate flexibility available (we go into more detail about these approaches in the related work section below).  However, to this point there are no real-time designs of the coordination between an aggregator and a system operator that achieve the goals laid out above. In particular, the goal of providing a \emph{real-time} feedback signal that is concise and informative has seemed unapproachable and so nearly all prior work has focused on slower-timescale estimations. In addition to having a flexibility feedback signal that is concise and informative, it is also desirable to have the feedback satisfy the following property:
\begin{enumerate}
\setcounter{enumi}{3}
   \item \textbf{Real-time.} The system is time-varying and non-stationary and so it is crucial that (nearly) real-time feedback can be defined and approximated if it is to be used in online feedback-based applications.
\end{enumerate}
The need for real-time information requires that computation of the feedback signal be simple and efficient, which is in direct conflict with assuring generality across wide-ranging applications.  In addition, it is highly desirable that the feedback signal be intuitive and interpretable, so that the ISO can use it at a policy level for planning purposes.

\textbf{Contributions.}  
In this paper we propose a novel design of a flexibility feedback signal that
quantifies the flexibility available to an aggregator.  We justify our design
by proving several desirable properties of the flexibility feedback for 
real-time feedback-based applications.  In particular our flexibility feedback
allows the system operator to maintain feasibility and enhance flexibility in
real time in an online setting.
Finally we demonstrate our design through two example applications: 
online cost minimization and real time capacity estimation.  
Our design is conceptually simple, interpretable and 
we describe two approximations that are efficiently 
computable in real time. Finally, it is the unique design that
attains a certain system capacity for flexibility in an
offline setting.


In more detail, we introduce a model of the real-time closed-loop control system formed by a system opeartor and an aggregator. Within this model we define the ``optimal'' real-time flexibility feedback vector as the solution to an optimization problem that maximizes the entropy of the feedback vector.  The use of entropy in this context is novel and we show that entropic maximization has a close relationship to maximization of the system capacity. Further, we justify 
axiomatically
how entropic maximization is fundamentally necessary for providing informative and concise feedback from the aggregator to the operator.

To illustrate applicability of the optimal real-time flexibility feedback vector we propose using two applications: online cost minimization and system capacity estimation. We demonstrate the effectiveness of the flexibility feedback vector in these applications through using real EV charging data from Caltech's ACN-Data dataset~\cite{lee2019acn}.  In the case of online cost minimization, we use the flexibility feedback signal in the context of model predictive control and show that the signal is effective even when it is approximated via a data-driven approach based on reinforcement learning.  In the case of system capacity estimation, we use the flexibility feedback signal in the context of Monte Carlo estimation and show that the signal is effective even when it is approximated via look-ahead estimation (rather than estimation based on historical data).  In both cases we provide provable guarantees that when the aggregator communicates with the system operator via the optimal flexibility feedback signal the private constraints of the loads governed by the aggregator are respected despite the conciseness of the signal communicated to the operator.  This work is the first to close the loop and both define a concise measure of aggregate flexibility and show how it can be used by the system operator to optimize system objectives while respecting the constraints of loads.  

\textbf{Related literature.} 
The growing importance of aggregators for the integration of controllable loads and the challenge of defining and quantifying the flexibility provided by aggregators means that a rich literature on the topic has emerged.  Broadly, this work can be separated into three approaches.  

\textit{Convex geometric approximation.} The idea of representing the set of aggregate loads as a virtual battery model dates back to~\cite{hao2014characterizing,hao2014aggregate}. In~\cite{zhao2017geometric}, flexibility of an aggregation of thermostatically controlled loads (TCLs) was defined as the Minkowski sum of individual polytopes, which is approximated by the homothets of a virtual battery model using linear programming. The recent paper~\cite{chen2018aggregating} takes a different approach and defines the aggregate flexibility as upper and lower bounds so that each trajectory to be tracked between the bounds is disaggregatable and thus feasible. However, convex geometric approaches cannot be extended to generate real-time flexibility signals because the approximated sets cannot be decomposed along the time axis. In~\cite{bernstein2016aggregation}, a belief function of setpoints is introduced for real-time control. However, feasibility can only be guaranteed when each setpoint is in the belief set and this may not be the case for systems with memory.

\textit{Scheduling algorithm-driven analysis.} Scheduling algorithms that enable the aggregation of loads have been studied in depth over the past decade. The authors of~\cite{gan2012optimal} introduced a decentralized algorithm with a real-time
implementation for EV charging to track a given load profile.
The authors of
\cite{subramanian2012real} considered the feasibility of matching a given power trajectory and show that causal optimal policies do not exist. In this work, aggregate flexibility was implicitly considered as the set of all feasible power trajectories. Three heuristic causal scheduling policies were compared and the results were extended to aggregation of deferrable loads and
storage in~\cite{subramanian2013real}. Furthermore, decentralized participation
of flexible demand from heat pumps and electric vehicles was
addressed in~\cite{papadaskalopoulos2013decentralized}.   Notably, the flexibility signals that have emerged from this literature are not general, i.e., the apply to specific policies and DERs.

\textit{Probability-based characterization.} There is much less work on probabilistic methods. The aggregate flexibility of residential loads was defined based on positive and negative
pattern variations by analyzing collective behaviour of aggregate users~\cite{sajjad2016definitions}. A randomized and decentralized control architecture for systems of deferrable loads was proposed in~\cite{meyn2015ancillary}, with a linear time-invariant system approximation of the derived aggregate nonlinear model. Flexibility in this work was defined as an estimate of the proportion of loads that are operating.  Our work falls into this category, but differs from previous papers in that entropy maximization for a closed-loop control system yield an interpretable signal that can be informative for operator objectives in real-time, as well as guarantee feasibility of the private constraints of loads.

\textit{Other approaches.} Beyond the works described above, there are many other suggestions for metrics of aggregate flexibility, e.g., graphical-based measures~\cite{kara2015estimating} and data-driven approaches~\cite{kara2015estimating}. 
Most of these, and the approaches described above are evaluated at the aggregator level however, and much less attention has been paid to the question of real-time coordination between an ISO and an aggregator that controls decentralized loads. 

The assessment and enhancement of aggregate flexibility are often considered independent of the operational objectives and constraints today. For instance, the notion of aggregated flexibility is reported to an ISO participating in a reserve market a day ahead and the scheduling is then conducted the next day after receiving the flexibility representation as defined in~\cite{hao2014characterizing,chen2018aggregate,chen2018aggregating,madjidian2018energy}, with notable exceptions, such as~\cite{wenzel2017real}, which considered charging and discharging of EV fleets batteries for tracking a sequence of automatic generation control (AGC) signals. However, this approach has several limitations. First, in large-scale systems, knowing the exact states of each load is not realistic. Second, classical flexibility representations often rely on a precise state-transition model on the aggregator's side. Third, traditional ISO market designs, such as a day-ahead energy market, often make use of ex ante estimates of future system states. The forecasts of the future states can sometime be far from reality, because of either an inaccurate model is used, or an uncertain event occurs. In contrast, a real-time energy market~\cite{marzband2013experimental,siano2016assessing} provides more robust system control when facing uncertainty in the  environment, e.g., from fast-changing renewable resources or human behavioral parameters. This further highlights the need for real-time flexibility feedback, and serves to differentiate the approach in our paper.

\textbf{Notation and Conventions. } 
We use $\mathbb{P}\left(\cdot\right)$ and $\mathbb{E}\left(\cdot \right)$ to denote the probability distribution and expectation of random variables. The (discrete) entropy function is denoted by $\mathbb{H}(\cdot)$. To distinguish random variables and their realizations, we follow the convention to denote the former by capital letters (e.g., $X$) and the latter by lower case letters (e.g., $x$). Furthermore, we denote the length-$t$ prefix of a vector $x$ by $x_{\leq t}:=(x_1,\ldots,x_{t})$. Similarly, $x_{<t}:=(x_1,\ldots,x_{t-1})$ and $x_{a\rightarrow b}:=(x_{a},\ldots,x_{b})$. The concatenation of two vectors $x$ and $y$ is denoted by $(x, y)$. Given two vectors $x,y\in\mathbb{R}^{\nn}$, we write $ x\preceq y$ if $x_i\leq y_i$ for all $i=1,\ldots,\nn$. For $x\in\mathbb{R}$, denote $[x]_+:=\max\{0,x\}$. 




\section{Problem Formulation}
\label{sec:model}



Consider a load aggregator and a system operator that interact over a
discrete time horizon $[T] := \{1, \dots, T\}$.

\subsection{Load aggregator}

Let $\xi_t$ denote the \emph{aggregator state} at time $t$ that takes value in a certain set $\Omega$.  Let $\xi_{\leq t} := (\xi_s \in\Omega : s = 1, \dots, t )$ denote the aggregator state trajectory up to time $t$ and $\xi: = \xi_{\leq T}$. 
The aggregator needs to accomplish a certain task over the horizon $[T]$, e.g., delivering energy to a set of electric vehicles (EVs) by their deadlines.
To this end, it makes a decision $\phi_t$ at each time $t$ according to a \emph{disaggregation policy} $\phi$. 
The decision $\phi_t$ changes the aggregator state $\xi_t$ according 
to a state transition function which is not essential for our 
discussion.  We hence omit its description and represent the dynamics of
the aggregator simply by the state trajectory $\xi$.
Besides accomplishing its task, the decision $\phi_t$ 
also produces a {system input} at time $t$ that will affect 
a system cost, e.g., the aggregate EV charging rate increases
load on the electricity grid.
The aggregator has flexibility in its decisions $\phi_t$ for accomplishing its task and, we assume for this paper, is indifferent to these decisions as long as the task is accomplished by time $T$.
At each time $t$ the system operator sends a \emph{signal} 
$x_t$ to the aggregator to guide the aggregator's decision $\phi_t$
towards one that minimizes the system cost.
The signal $x_{t}$ at time $t\in[\nt]$ takes value in a discrete set $\mathbb{X}\subseteq\mathbb{R}$.\footnote{We assume that the set $\mathbb{X}$ is discrete only for simplicity of presentation. Our results, for example, the definition of optimal flexibility feedback (Definition~\ref{def:flexibility_optimal}), Theorem~\ref{thm:con} can be extended to continuous space using a density function as the flexibility feedback, changing the summations to integrals, replacing the discrete entropy functions by differential entropy functions, and redefining the system capacity $\log |\mathcal S(\phi, \xi)|$ as the volume of 
the space consisting of feasible signal trajectories.}
Let $x_{\leq t}:=(x_1,\ldots,x_{t})$ denote the signal trajectory up to
time $t$ and $x := x_{\leq T}$.
In general, the aggregator's decision $\phi_t := \phi_t(\xi_t, x_{\leq t})$ is a causal function of aggregator state $\xi_t$ and signal trajectory
$x_{\leq t}$ up to time $t$.\footnote{The main results can be easily extended to allow for non-causal policies.}  
We use $\phi$ both to denote the disaggregation policy or
the decision trajectory $$\phi=\left( \phi_1(\xi_{1},x_{1}),\ldots,\phi_{\nt}(\xi_T,x)\right)$$
depending on the context.
We often refer to a pair of disaggregation policy and aggregator state
trajectories $(\phi, \xi)$ as an \emph{aggregator trajectory}.
That the aggregator must accomplish its task but is otherwise indifferent 
to its decisions $\phi$ can be modeled by the constraints:
\begin{align}
\label{eq:2.2}
g_i\left(x;\phi,\xi\right) &\leq 0, \ i=1,\ldots, m,
\end{align}
where each $g_i$ is an arbitrary function of $\phi$, $\xi$ and $x$. 



The disaggregation policy $\phi$ can represent a variety of control strategies, such as a scheduling algorithm for EV charging, energy disaggregation, or price signals.
We illustrate our model of an aggregator using an EV charging application. 

\begin{example}[Aggregator: EV charging]
\label{example:ev}
\emph{Consider an aggregator that is an EV charging facility with $\nn$ users. 
Each user $j$ has a private vector $\left(a(j), d(j),e(j),r(j)\right)\in\mathbb{R}^4$ where $a(j)$ denotes its arrival (connecting) time; $d(j)$ denotes its departure (disconnecting) time, normalized according to the time indices in $[\nt]$; $e(j)$ denotes the total energy to be delivered, and $r(j)$ is its peak charging rate. 
Fix a set of $\nn$ users with their private vectors $\left(a(j), d(j),e(j),r(j)\right)$, 
the aggregator state $\xi_t$ at time $t\in [\nt]$ is a collection of
length-3 vectors $(d(j), e_t(j), r(j) : a(j) \leq t \leq d(j))$ for each EV that
has arrived and has not departed by time $t$.  
Here $e_t(j)$ is the remaining energy 
demand of user $j$ at time $t$.
The decision $\phi_t(j)$ is the energy delivered to each user $j$ at time $t$.
A policy $\phi$ can be well-known scheduling policies such as earliest-deadline-first, least-laxity-first, etc.
The aggregator decision $\phi_t(j) \in{\mathbb{R}}_{+}$ at each time $t$ 
updates the state, in particular $e_t(j)$, even though do not explicitly
represent the state transition function.  The decision also produces an 
aggregate charging energy
$\sum_{j: a(j) \leq t \leq d(j)} \phi_t(j)$ that affects the load on the
power grid and operational cost.
}

\emph{Suppose, in the context of demand response, the system operator 
(a local utility company, or a building management) sends
a signal $x_t$ that is the aggregate energy that can be allocated to EV charging.
The aggregator makes charging decisions $\phi_t(j)$ to track the signal $x_t$ received from the system operator as long as they will meet the energy demands 
of all users before their deadlines. 
Then the constraints in \eqref{eq:2.2} include the following constraints on 
the charging decisions:}
\begin{subequations}
\begin{align}
\nonumber
\phi_{t}(j)= 0 \ , \  t<a(j), & \ j=1,\ldots,\nn,\\
\nonumber
\phi_{t}(j)= 0\ , \ t>d(j), & \  j=1,\ldots,\nn,\\
\label{eq:f3}
\sum_{j=1}^{\nn}\phi_{t}(j)= x_{t}, & \ t=1,\ldots,\nt,\\
\label{eq:f2}
\sum_{t=1}^{\nt}\phi_{t}(j) = e(j), & \ j=1,\ldots,\nn,\\
0\leq\phi_{t}(j)\leq r(j), & \ t=1,\ldots,\nt
\end{align}
\label{eq:f123}
\end{subequations}
\emph{where constraint \eqref{eq:f3} ensures that the aggregator decision $\phi_t$ tracks the signal $x_{t}$ at each time $t\in [\nt]$,
the constraint (\ref{eq:f2}) guarantees that EV $j$'s energy demand is satisfied,  and the other constraints say that the aggregator cannot charge an EV before its arrival, after its departure, or at a rate that exceeds its limit.}
\end{example}

\subsection{System operator}
As Example \ref{example:ev} illustrates, the aggregator decisions $\phi_t$ 
produce a system input that affects system operation.
The goal of the system operator is to compute a 
signal $x_t$ at time $t\in [\nt]$ to guide the aggregator's decisions
$\phi_t$ so as to minimize the {system cost} given by a \emph{cost function} 
 $f:\mathbb{X}^{\nt}\rightarrow\mathbb{R}$.
 The signal trajectory $x$ must satisfy certain operational constraints,  parameterized by an \textit{environmental parameter}  $\zeta$:
\begin{align}
\label{eq:2.3}
h_i\left(x; \zeta\right) &\leq 0, \ i=1,\ldots, k.
\end{align}

\begin{example}[Operational constraints]
\label{example:operational_constraints}
\emph{Suppose $x_t$ represents the total load of an EV charging facility, or on a power system.
If the operator performs peak shaving, 
then \eqref{eq:2.3} may be:}
\begin{align}
\label{eq:max_rate}
   x_{t}\leq \overline{\gamma}, \ t=1,\ldots,\nt.
\end{align}
\emph{If the operator limits ramp rates, 
then \eqref{eq:2.3} may be:}
\begin{align}
\label{eq:ramp}
   \left|x_{t+1} - x_{t}\right|\leq \varepsilon, \ t=1,\ldots,\nt.
\end{align}
\end{example}


\begin{example}[Cost function]
\label{example:online_optmization}
\emph{Suppose the electricity cost at each time $t\in [\nt]$ is a function $f_t:\mathbb{X}\rightarrow \mathbb{R}_{+}$.  Then the total electricity cost is $ f(x) := \sum_{t=1}^{\nt}f_t(x_t)$.}


%
\end{example}

The goal of the system operator is to choose the signal $x$ so as to solve:
\begin{subequations}
\begin{align}
\label{eq:offline_1}
\min_{x} f(x)\\
\label{eq:offline_2}
\text{subject to } g_i\left(x; \phi,\xi\right) &\leq 0, \ i=1,\ldots, m,\\
\label{eq:offline_3}
h_i\left(x; \zeta\right) &\leq 0, \ i=1,\ldots, k.
\end{align}
i.e., 
\label{eq:offline}
\end{subequations}
the operator 
wishes to minimize its cost $f$ subject to its operational constraints  \eqref{eq:offline_3}  while the load aggregator 
needs to fulfill its obligations in the form of constraints 
\eqref{eq:offline_2}.  
This is an offline problem that involves global information at all
times $t\in [\nt]$.
The challenge is that the constraints \eqref{eq:offline_2} are 
private to the aggregator.  
It is impractical for the aggregator to 
communicate the constraint functions $g_i\left(x; \phi,\xi\right)$
to the operator because of privacy concerns or computational effort, and because in an online setting, even the aggregator will not know all the
constraints at each time $t$ that involve future information, e.g., future EV arrivals in Example \ref{example:ev}.

\begin{remark}
\label{remark:online}
For simplicity, we describe our model in an offline setting where 
the cost and the constraints 
(e.g., see \eqref{eq:f2} in Example 2.1) in the optimizatoin
problem \eqref{eq:offline} 
are expressed in terms of the entire trajectories $(x; \phi, \xi)$.
All functions defined in this paper, however, are \emph{causal} in 
that they depend \emph{only 
on local information available at time $t$}.  Hence these functions
are designed for solving an online version of the offline problem
\eqref{eq:offline}.
\end{remark}

\subsection{Online feedback-based solution}
\label{sec:optimization}

We explore a solution where the system operator and the 
aggregator jointly solve an online version of \eqref{eq:offline} in
a closed loop in real time, as illustrated in Figure \ref{fig:framework}.
\begin{figure*}[htbp]
	\centering
	\includegraphics[scale=0.35]{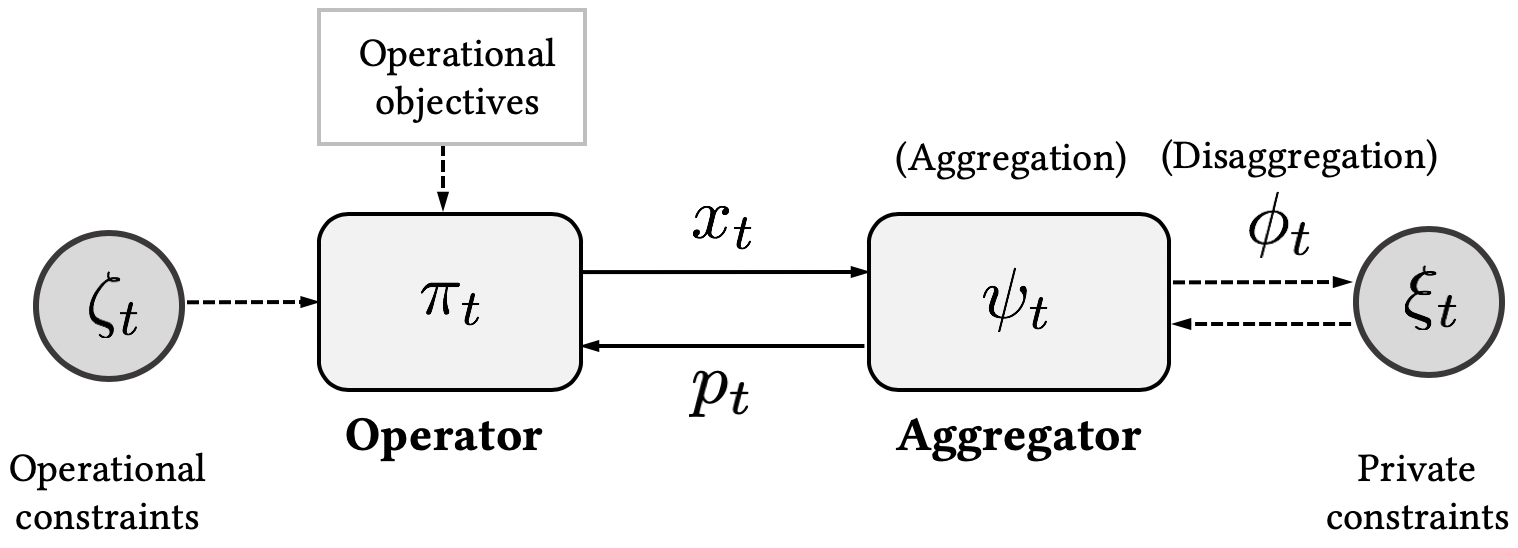}
	\caption{{A feedback control approach for solving an online 
	version of \eqref{eq:offline}. }}
	\label{fig:framework}
	\medskip
\end{figure*}
Our approach does not require the aggregator to know the system 
operator's optimization problem \eqref{eq:offline}, but only the
signal $x_t$ at each time $t$ from the operator.
It does not require the system operator to know the aggregator 
constraints \eqref{eq:offline_2}, but only a feedback signal 
$p_t$ (to be designed) from the aggregator.  The system
operator generates its signal $x_t$ using a causal function $\pi_t$
and the aggregator generates its feedback $p_t$ using a
causal function $\psi_t$.  By an ``online feedback'' solution, we mean
that these functions $(\pi_t, \psi_t)$ use {only information available 
locally at time $t$}.

Specifically, our approach proceeds as follows.  At each time $t$, the
aggregator computes a length-$|\mathbb{X}|$ vector 
\begin{subequations}
\begin{eqnarray}
p_t(\cdot|x_{<t}; \xi_t) & = & \psi_t (x_{<t}; \xi_t) \ \ =: \ \  \psi_t (x_{<t})
\label{eq:feedback1.pt}
\end{eqnarray}
based on its current state $\xi_t$ and previously received signal
trajectory ${x}_{<t}=(x_1,\ldots,x_{t-1})$, and sends it to
the system operator.  We will omit $\xi_t$ in the notation when it is not 
essential to our discussion and simplify the probability vector as $p_t$.\footnote{Note that in~\eqref{eq:af3} we slightly abuse the notation and use $p_t$ to denote a conditional distribution. This is only for computational purposes and the information sent from an aggregator to an operator at time $t\in [\nt]$ is still a length-$|\mathbb{X}|$ probability vector, conditioned on fixed $x_{<t}$.} 
The system operator then computes a (possibly random) signal
\begin{eqnarray}
x_t & = & \pi_t(p_t; \zeta) \ \ =: \ \ \pi_t(p_t)
\label{eq:feedback1.xt}
\end{eqnarray}
based on the aggregator feedback $p_t$ and sends it to the aggregator.
\label{eq:feedback1}
\end{subequations}
 We will omit $\zeta$ in the notation when it is not essential to our discussion. 
The aggregator makes its decision $\phi_t(\xi_t, x_{\leq t})$.
It then computes the next feedback $p_{t+1}$ and the cycle repeats.

The operator chooses its signal $x_t$ in order to solve the time-$t$ problem
in an online version of \eqref{eq:offline}, so the function $\pi_t$ denotes
the mapping from the aggregator feedback $p_t$ to an optimal solution
of the time-$t$ problem.  See Section \ref{sec:online} for an example.

The focus of this paper is to propose an aggregator feedback $\psi_t$ in \eqref{eq:feedback1.pt}
that quantifies its future flexibility that will be enabled by an operator 
decision $x_t$.  The feedback $p_t$ therefore is a surrogate for the
aggregator constraints \eqref{eq:offline_2} to guide the operator's decision.
Specifically, define the set of all \textit{feasible signal trajectories} for the aggregator 
as:
\begin{align*}
\mathcal{S}(\phi,\xi):=\left\{x\in\mathbb{X}^{\nt}: x
    \text{ satisfies } \eqref{eq:offline_2} \right\}.
\end{align*}
Throughout, we assume that $\mathcal{S}(\phi,\xi)$ is non-empty. 
We propose that the aggregator function $\psi_t(x_{<t}; \xi_t)$ 
computes the conditional probabilities of future signal trajectories 
$x_{>t} := (x_{t+1}, \dots, x_T)$ that satisfy
the aggregator constraints \eqref{eq:offline_2}, as a function of the operator's
signal choice $x_t$, conditioned on the signal trajectory $x_{<t} := (x_1, \dots, x_{t-1})$
up to time $t-1$.
Formally, let $\mathcal P$ denote the probability simplex:
\begin{align*}
\mathcal{P}:=\left\{p\in\mathbb{R}^{|\mathbb{X}|}:p(x)\geq 0, x\in\mathbb{X}; \sum_{x\in\mathbb{X}}p(x)=1\right\}.
\end{align*} 
Fix any aggregator trajectory $(\phi, \xi)$.
Then the aggregator function $\psi_t: \mathbb{X}^{t-1} \times \Omega \rightarrow\mathcal{P}$ at each time $t$ is:  
$p_t \ = \ \psi_t(x_{<t}; \xi_t)$ such that for each $x_t\in \mathbb X$,
\begin{align}
\psi_t(x_{<t}; \xi_t) \ := \
p\left(\cdot | (x_{<t}) \right)\in\mathcal{P}.
\label{eq:psit.1}
\end{align}
We refer to $p_t$ as \emph{flexibility feedback} sent at time 
$t\in [\nt]$ from the aggregator to the system operator.
Given current aggregator state $\xi_t$ and signal trajectory $x_{<t}$,
the conditional probability $\psi_t(x_{<t}; \xi_t)$ depends not just on
the operator decision $x_t$, but also on the future evolution of the aggregator
state $\xi_t$. 
In this paper, we do not fully specify the details of
the dynamical process $\xi_t$. For different applications, $\xi_t$ may evolve
according to different state transition functions, possibly with stochastic inputs.
These details will determine the value of the flexibility feedback 
$p_t = \psi_t(x_{<t}; \xi_t)$ defined in \eqref{eq:psit.1}.

In this sense, \eqref{eq:psit.1} does not specify a specific
aggregator function $\psi_t$, but a class of possible functions
$\psi_t$.  Every function in this collection is \emph{causal} in that it
depends only on information available to the aggregator at time
$t$.  
In contrast to most aggregate flexibility notions in the literature~\cite{hao2014characterizing,hao2014aggregate,sajjad2016definitions,zhao2017geometric,madjidian2018energy,chen2018aggregating,sadeghianpourhamami2018quantitive,evans2019graphical}, the 
flexibility feedback  here 
is specifically designed for an online feedback control setting.

\section{Optimal Flexibility Feedback}
\label{sec:real_time}
\label{sec:feedback}

In this section we propose a specific function $\psi_t$
in the class defined by \eqref{eq:psit.1} for computing
aggregator feedback to quantify its future flexibility.  
We will justify our proposal by showing that the proposed $\psi_t$ 
has several desirable properties for solving an online 
version of (6) using the real-time feedback-based approach 
\eqref{eq:feedback1}.

\subsection{Definition}

The intuition behind our proposal is that the conditional
probability $p_t(x_t) := p_t(x_t|x_{<t})$ measures 
the resulting future flexibility of the aggregator
if the system operator chooses $x_t$ as the signal at 
time $t$, given the signal trajectory up to time $t-1$.
The sum of the conditional entropy of $p_t$ thus is a 
measure of how informative $p_t$ is.
This suggests choosing a conditional distribution $p_t$ that 
maximizes its conditional entropy.
Fix any aggregator trajectory $(\phi, \xi)$.  
Consider the optimization problem:
\begin{subequations}
\begin{align}
\label{eq:af1}
{\digamma}(\phi,\xi) \ := \ \max_{p_1,\ldots,p_{\nt}}\ \sum_{t=1}^{\nt}\mathbb{H}\left({X}_t|X_{<t}\right)\ 
    \text{subject to} \ X\in\mathcal{S}(\phi,\xi)
\end{align}
where the variables are conditional distributions:
\begin{align}
p_t & \ := \ p_t(\cdot|\cdot):=\mathbb{P}_{X_t|X_{<t}}(\cdot|\cdot),
\qquad t \in [\nt]
\label{eq:af3}
\end{align}
$X\in \mathbb X^T$ is a random variable distributed according to the joint distribution $\prod_{t=1}^{\nt}p_t$ and $\mathbb{H}\left({X}_t|X_{<t}\right)$ is the conditional entropy of  $p_t$ defined as:
\begin{align}
\label{eq:ent}
\mathbb{H}\left({X}_t|X_{<t} \right) & \ := \
\sum_{x_1,\ldots,x_{t}\in\mathbb{X}}\Big(-\prod_{\ell=1}^{t}p_\ell(x_\ell|x_{<\ell})\Big)\log{p_t(x_{t}|x_{<t})}.
\end{align}
By definition, a quantity conditioned on ``$x_{<1}$'' means an unconditional 
quantity, so in the above, $\mathbb{H}\left({X}_1|X_{<1} \right) := 
\mathbb{H}\left({X}_1\right) := \mathbb{H}\left({p}_1\right)$.
The chain rule shows that $\sum_{t=1}^{\nt}\mathbb{H}\left({X}_t|X_{<t}\right) = \mathbb{H}\left(X\right)$.
Hence \eqref{eq:af} can be interpreted as maximizing the entropy
$\mathbb{H}\left(X\right)$ of a random trajectory $X$ sampled according 
to the joint distribution $\prod_{t=1}^{\nt}p_t$, conditioned on $X$ satisfying~\eqref{eq:2.2}, where the maximization is over the collection
of conditional distributions $(p_1, \dots, p_T)$.
We provide in Section \ref{app:axiom}
an axiomatic justification of maximizing the entropy 
$\mathbb{H}\left(X\right)$ of the signal trajectory $X$ in \eqref{eq:af1}.

\label{eq:af}
\end{subequations}
\begin{definition}[Optimal flexibility feedback]
\label{def:flexibility_optimal}
Fix any aggregator trajectory $(\phi, \xi)$.
The flexibility feedback $p_t^* = \psi^*_t(x_{<t}; \xi_t)$ for $t\in [\nt]$
is called the \textbf{optimal flexibility feedback} if 
$(p_1^*, \dots, p_T^*)$ is the unique optimal solution of \eqref{eq:af}.
\end{definition}

\begin{remark}
Even though the optimization problem \eqref{eq:af} involves variables $p_t$
for the entire time horizon $[\nt]$, the individual variables
$p_t$ in \eqref{eq:af3} are conditional probabilities that depend
only on information available to the aggregator at times $t$.
Therefore the optimal flexibility feedback $\psi^*_t$ in Definition \ref{def:flexibility_optimal}
is indeed causal and in the class of functions $\psi^*_t$ defined in \eqref{eq:psit.1}.
The existence and uniqueness of $p^*_t$ is guaranteed by Theorem \ref{thm:con}
below, which also implies that $\psi^*_t$ is unique.
\qed\end{remark}

We demonstrate Definition \ref{def:flexibility_optimal} using a toy example.
\begin{example}[Optimal flexibility feedback $p^*$]
\label{example:toy}
Consider the following instance of Example~\ref{example:ev}. Suppose the number of charging time slots is $\nt=3$ and there is one customer, whose private vector is $(1,3,1,1)$ and possible energy levels are $0$ (kWh) and $1$ (kWh), i.e., $\mathbb{X}\equiv\{0,1\}$. Since there is only one EV, the scheduling algorithm $\phi$ (disaggregation policy) assigns all power to this single EV. For this particular choices of $\xi$ and $\phi$, the set of feasible trajectories is $\mathcal{S}(\phi,\xi)=\{(0,0,1),(0,1,0), (1,0,0)\}$, shown in Figure~\ref{fig:toy_example} with the corresponding optimal conditional distributions given by~\eqref{eq:af}. 
\begin{figure}[h]
    \centering
    \includegraphics[scale=0.29]{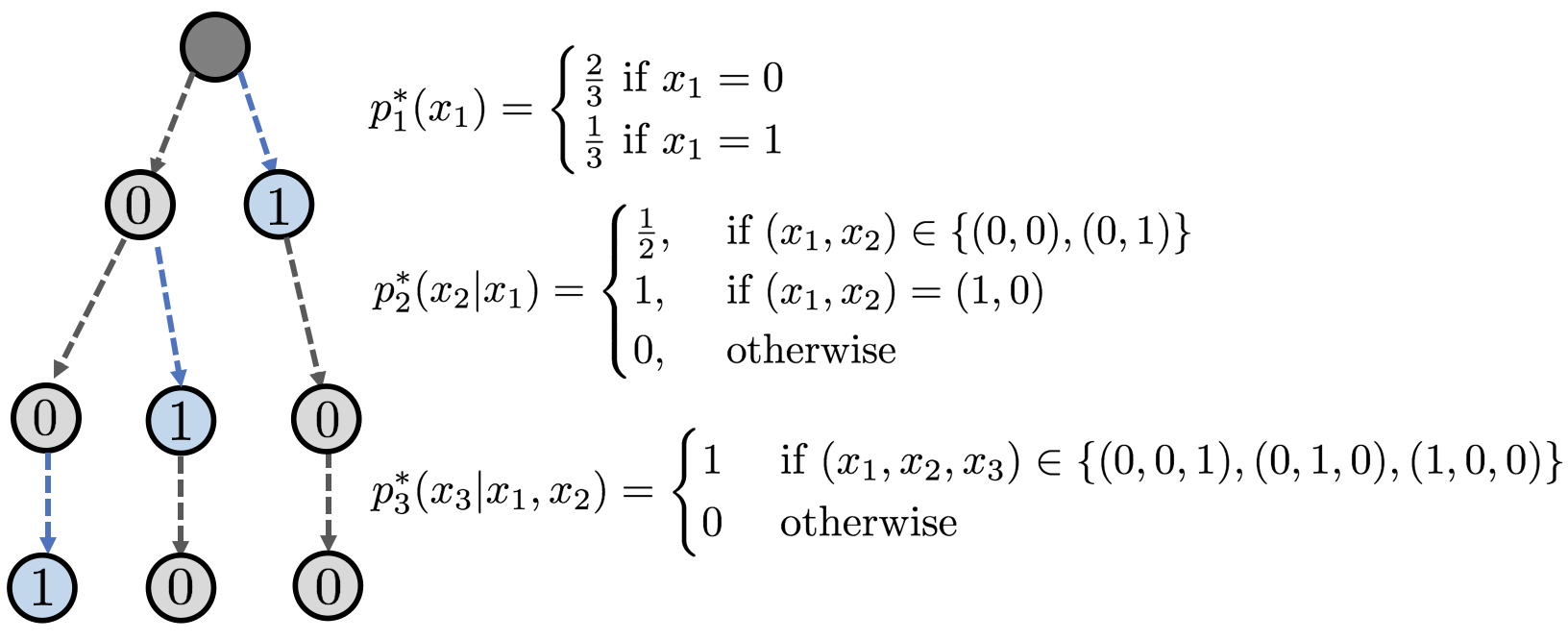}
    \caption{Feasible trajectories of power signals and the computed optimal flexibility feedback in Example~\ref{example:toy}.}
    \label{fig:toy_example}
\end{figure}
\end{example}

\subsection{Properties of $p^*_t$}
\label{subsec:properties}

We now show that the proposed optimal flexibility feedback $p^*_t$ has several desirable properties.
We start by computing $p^*_t$ explicitly. 
Fix any aggregator trajectory $(\phi, \xi)$.
Given any signal trajectory $x_{\leq t}$, define the 
set of \emph{subsequent} feasible trajectories as:
\begin{align}
\mathcal{S}(\phi,\xi|x_{\leq t})\ := \ \Big\{x_{>t}\in\mathbb{X}^{\nt-t}:  g_i\left(\phi,\xi,x\right) \leq 0, \
 \forall i=1,\ldots, m\Big\}.
\label{eq:set_of_feasible}
\end{align}
where $x := (x_{\leq t}, x_{>t})$.  The size 
$\left|\mathcal{S}(\phi,\xi|x_{\leq t})\right|$ of the set of subsequent
feasible trajectories is a measure of future flexibility, conditioned on
$x_{\leq t}$. Our first result justifies our calling $p^*_t$ the optimal
flexibility feedback: $p^*_t$ is a measure of the future flexibility that will be
enabled by the operator's signal $x_t$ and it attains a measure of
system capacity for flexibility (see Remark \ref{remark:systemcap} below).
By definition, 
$\mathcal{S}(\phi,\xi|x_{<1})\ := \ \mathcal{S}(\phi,\xi)$ and 
$p_1^*(x_1|x_{<1}) \ := \ p_1^*(x_1)$. 
\begin{theorem}
\label{thm:con}
The {optimal} flexibility feedback $p^*_t$ is given by
\begin{align}
\label{eq:optimal}
p_t^*(x_t|x_{<t}) \ = \ \frac{\left|\mathcal{S} \left(\phi,\xi | 
    (x_{< t}, x_t) \right) \right|}{\left|\mathcal{S}(\phi,\xi|x_{< t})\right|},  
    \quad  \forall (x_{< t}, x_t) \in\mathbb{X}^{t}.
\end{align}
for $t \in [\nt]$.
Moreover, the optimal value ${\digamma}(\phi,\xi)$ of \eqref{eq:af} is equal to $\log\left|\mathcal{S}(\phi,\xi)\right|$.
\end{theorem}

\begin{proof} We prove the statement by induction on $T$. It is straightforward to verify the results when $\nt=1$. We suppose the theorem is true when $\nt=m$. Suppose $\nt=m+1$.
Let
\begin{align*}
\digamma(\phi,\xi|x_1):=\max_{p_2,\ldots,p_{\nt}}\sum_{t=2}^{\nt}\mathbb{H}\left(X_{t}|X_{<t}\right)
\end{align*}
denote the optimal value corresponding to the time horizon 
$t\in [\nt]\backslash \{1\}$, conditioning on $x_1$. We have
\begin{align*}
{\digamma}(\phi,\xi) = \max_{p_{1}}\sum_{x_1\in\mathbb{X}}
    p_{1}(x_1)\digamma(\phi,\xi|x_1)+\mathbb{H}(p_{1}).
\end{align*}
By the induction hypothesis, $\digamma(\phi,\xi|x_1)= \log\left|\mathcal{S}(\phi,\xi|x_1)\right|$. Therefore, 
\begin{align*}
{\digamma}(\phi,\xi) =& \max_{p_{1}}\sum_{x_1\in\mathbb{X}}p_{1}(x_1)\log\left|\mathcal{S}(\phi,\xi|x_1)\right|+\mathbb{H}(p_{1})\\
=&\max_{p_{1}}\sum_{x_1\in\mathbb{X}}p_{1}(x_1)\log\left(\frac{\left|\mathcal{S}(\phi,\xi|x_1)\right|}{p_1(x_1)}\right)
\end{align*}
whose optimizer $p_1^*$ satisfies (\ref{eq:optimal}) and we get ${\digamma}(\phi,\xi)=\log\left|\mathcal{S}(\phi,\xi)\right|$. The theorem follows by finding the optimal conditional distributions $p_2^*,\ldots,p_{\nt}^*$ inductively.
\end{proof}

Given the unique optimal flexibility feedback $(p^*_1,\ldots,p^*_{\nt})$ guaranteed
by Theorem \ref{thm:con}, let
$q^*(x)=\prod_{t=1}^{\nt}p^*_t(x_{t}|x_{<t})$ denote the joint distribution
of the signal trajectory $x$.  Then \eqref{eq:optimal} implies that the joint
distribution $q^*$ is the uniform distribution over the set $\mathcal{S}\left(\phi,\xi \right)$ of all feasible trajectories:
\begin{align}
\label{eq:uniform}
    q^*(x):=\begin{cases}
    1/\left|\mathcal{S}(\phi,\xi)\right| & \text{ if } x\in \mathcal{S}(\phi,\xi)\\
    0 & \text{ otherwise}
    \end{cases}.
\end{align}
\begin{remark}[System capacity ${\digamma}(\phi,\xi)$]
\label{remark:systemcap}
Fix any aggregator trajectory $(\phi, \xi)$.
The size $\left|\mathcal{S}\left(\phi,\xi\right)\right|$ 
is a measure of flexibility inherent in the aggregator.  
We will hence call
$\log \left|\mathcal{S}\left(\phi,\xi\right)\right|$ the \emph{system capacity}.
Theorem \ref{thm:con} then says that the optimal value of~\eqref{eq:af}
is the system capacity, 
${\digamma}(\phi,\xi)=\log \left|\mathcal{S}\left(\phi,\xi\right)\right|$.
Moreover
the optimal flexibility feedback $(p^*_1, \dots, p^*_T)$ is the unique 
collection of conditional distributions that attains the system capacity
in ~\eqref{eq:af}.
This is intuitive since the entropy of a random trajectory $x$ in $\mathcal{S}(\phi,\xi)$ is maximized by the uniform distribution
$q^*$ in ~\eqref{eq:uniform} 
induced by the conditional distributions $(p_1^*,\ldots,p_{\nt}^*)$.
\qed
\end{remark}

Theorem \ref{thm:con} directly implies the following important properties of 
the optimal flexibility feedback.
\begin{corollary}[feasibility and flexibility]
\label{coro:property}
Let $p^*_t = p^*_t(\cdot|x_{<t})$ be the optimal flexibility feedback  
at each time $t\in [\nt]$.
\begin{enumerate}
\item 
For any signal trajectory $x=(x_1,\ldots,x_{\nt})$, if 
\begin{align*}
    p^*_t(x_t|x_{<t}) \ > \ 0 \quad \text{ for all } t\in [\nt]
\end{align*}
then $x\in\mathcal{S}(\phi,\xi)$. 

\item 
For all $x_t,x_t'\in\mathbb{X}$ at each time $t$, if
\begin{align*}
    p^*_t(x_t|x_{<t})\ \geq \ p^*_t(x_t'|x_{<t})
\end{align*}
then 
$|\mathcal{S}(\phi,\xi| (x_{<t}, x_t))|\geq |\mathcal{S}(\phi,\xi|(x_{<t}  ,x_t'))|$.
\end{enumerate}
\end{corollary}
We elaborate on the implication of Corollary \ref{coro:property}
on our online feedback-based solution approach.
\begin{remark}[Feasibility and flexibility]
\label{remark:feasibility}
Corollary \ref{coro:property} says that the proposed optimal flexibility
feedback $p^*_t$ provides the right information for the system operator to choose
its signal $x_t$ at time $t$.
Specifically, the first statement of the corollary says that if the operator 
always chooses a signal $x_t$ with positive conditional probability 
$p^*_t(x_t)>0$ for each time $t$, then the resulting signal trajectory is 
guaranteed to be feasible, $x\in\mathcal{S}(\phi,\xi)$, i.e., the system
will remain feasible at \emph{every} time $t$ along the way.  

Moreover, according to the second statement of the corollary, if the system operator
chooses a signal $x_t$ with a larger $p^*_t(x_t)$ value at time $t$, then 
the system will be more flexible going forward than if it had chosen another
signal $x_t'$ with a smaller $p^*_t(x_t')$ value, in the sense that there 
are more feasible trajectories in $\mathcal{S}(\phi,\xi| (x_{<t}, x_t))$ 
going forward.
\qed
\end{remark}
As noted in Remark \ref{remark:online}, despite characterizations 
that involve the whole trajectory
$(x, \phi, \xi)$, such as $x\in\mathcal{S}(\phi,\xi)$, these are \emph{online} properties.
This guarantees the feasibility of the online closed-loop control system 
depicted in Figure \ref{fig:framework}, and confirms the suitability 
of $p^*_t$ for online applications.

\subsection{Axiomatic justification of \eqref{eq:af}}
\label{app:axiom}

As explained in Remark \ref{remark:feasibility}, the optimal flexibility feedback
$p^*_t$ quantifies succinctly for the system operator the future flexibility of the
aggregator that will be enabled by the operator's choice of next signal $x_t$. 
Intuitively, the system has ``more flexibility'' at time $t$ if the distribution $p_t(\cdot|x_{<t})$ is ``more uniform''. 
This view suggests using an entropic measure to quantify flexibility, such 
as the cost function of the optimization problem \eqref{eq:af} that
underlies our proposed flexibility feedback.
In this subsection we justify this intuition using an axiomatic argument. 

Fix any aggregator trajectory $(\phi, \xi)$.
Consider a flexibility metric as a function of any flexibility feedback
$p\in \{p_1,\ldots,p_T\}$. Recall that $p$ is a conditional distribution. For any $p$, let 
$\digamma(p)$ represent a candidate metric for quantifying
aggregate flexibility.  Consider any time slots $\tau\in [\nt]$,  the metric should also be able to provide a value, given the marginal distributions $\overline{p}_{t}:=\sum_{\mathbf{x}_{<t}}p_{t}(\cdot|\mathbf{x}_{<t})\prod_{\tau<t}p_{\tau}(x_\tau|x_{<\tau})$.

%
We require the metric $\digamma$ to satisfy several conditions (axioms):
\begin{enumerate}

\item \emph{Continuity}: $\digamma(\overline{p}_t)$ is a continuous function of $\overline{p}_{t}$, $t\in [\nt]$.

\item \emph{(Strong) additivity}: $\digamma(q) = \sum_{t=1}^{\nt}\digamma(p_t)$
    if $q :=\prod_{t=1}^{\nt}p_t$. 

\item \emph{Subadditivity}: $\digamma(q_{t,t'}) \leq \digamma(\overline{p}_{t}) + \digamma(\overline{p}_{t'})$
    where $\overline{p}_{t}, \overline{p}_{t'}$ are marginal distributions corresponding to time slots $t$ and $t'$ and
    $q_{t,t'}$ is their joint distribution.

\item \emph{Symmetry}: $\digamma(q_{t,t'}) = \digamma(q_{t',t})$ where  $q_{t,t'}$  and  $q_{t',t}$ are joint distributions of time slots $t$ and $t'$.

\item \emph{Expansibility}: $\digamma(\overline{p}'_{t})=\digamma(\overline{p}'_{t})$ for all $\overline{p}_{t}$, $t\in [\nt]$ where $\overline{p}'_{t}=(\overline{p}_{t},0)$, i.e., concatenate a zero entry to $\overline{p}_{t}$. 

\end{enumerate}
Additivity is useful because the tracking of a random signal trajectory
$x := (x_1, \dots, x_T)$ can then be decomposed using the chain rule into sub-problems 
of tracking each signal $x_{t}$ at time $t$, conditioned on previous
signal trajectory $x_{<t}$. 
%
Subadditivity is motivated by the property that fixing a signal $x_t$ may restrict
the choice of feasible signals $x_{t'}$ 
since the signals $x_1,\ldots,x_\nt$ may be correlated.
This means that measuring the joint distribution of $(x_{t}, x_{t'})$ gives
lower flexibility than measuring the coordinates $x_{t}$ and $x_{t'}$ independently. 
%
%
For symmetry, the permutation of components in the distribution $\overline{p}_{t}$ does not change $\digamma(\overline{p}_{t})$ since the switch of positions does not affect the underlying distribution.
Expansibility is natural since adding a new component that equals to zero means $x_t$ can never choose a certain power level. So the aggregate flexibility will not change.

These five conditions imply that the flexibility metric $\digamma(\overline{p}_{t})$ (for all $t\in [\nt]$) must be an entropy function:
\begin{align*}
\mathbb{H}(\overline{p}_{t})&:= \sum_{x\in\mathbb{X}}\sum_{\mathbf{x}_{<t}}p_{t}(x|\mathbf{x}_{<t})\prod_{\tau<t}p_{\tau}(x_\tau|x_{<\tau})\\
&\cdot\log\left(\frac{1}{\sum_{\mathbf{x}_{<t}}p_{t}(x|\mathbf{x}_{<t})\prod_{\tau<t}p_{\tau}(x_\tau|x_{<\tau})}\right)
\end{align*} 
up to multiplicative factors and $\digamma(p_t)$ is the conditional entropy of $p_t$. This is a classical result about entropy; see ~\cite{csiszar2008axiomatic,aczel1974shannon}. 

The results in this section justify the design of using the unique 
optimal solution of ~\eqref{eq:af} as our flexibility feedback $p^*_t$.
The design attains the system capacity ${\digamma}(\phi,\xi)$. Moreover
it characterizes the aggregate flexibility in real-time and 
allows a decomposition (see Section \ref{sec:estimation} for details) of aggregate flexibility over $t$ via
\begin{align*}
    \sum_{t=1}^{\nt}\mathbb{H}\left(p^*_t\right) = {\digamma}(\phi,\xi).
\end{align*}
We use this decomposition in Section~\ref{sec:online} for online cost minimization 
where $p^*_t$ is used as a penalty in a RHC-based online algorithm.
We also use it in Section \ref{sec:real_time_eval} for estimating the system
capacity ${\digamma}(\phi,\xi)$ empirically using a Monte Carlo method.
Finally, computing the optimal flexibility feedback is demanding.  
We provide two approximations for $p^*_t$, one for the case where
sufficient historical data is available and the other when it is not. The first is a data-driven approach using reinforcement learning (Section~\ref{sec:RL}) and the second is a look-ahead approximation (Section~\ref{sec:look-ahead}).





\section{Online cost minimization}
\label{sec:online}

Consider the cost minimization problem introduced in Example~\ref{example:online_optmization}.  In this setting, the operator seeks to minimize the cost in an online manner, i.e., at time $t$ the operator only knows the objective functions $f_1,\ldots,f_t$ and the flexibility feedback $p_1,\ldots,p_t$.

We first describe a receding horizon control scheme for the operator that, given the flexibility feedback and the objective functions, allows the operator to compute the signals $x_1,\ldots,x_{\nt}$.  Then, we introduce a deep reinforcement learning-based approach for the aggregator to compute an approximation of the optimal flexibility feedback.  Finally, we illustrate our method with simulations.

\subsection{Operator: Receding horizon control}

The task of the operator is to, given the optimal flexibility feedback, generate signals $x_1,\ldots,x_{\nt}$ that are always feasible with respect to both the sets of private and operational constraints \emph{and} that minimize cost.  For the objective of cost minimization, we propose an approach that uses receding horizon control (RHC) to achieve this in an adaptive, online manner -- see Algorithm \ref{alg:online}.

\begin{algorithm}[t!]
\small
	\hrule
	\hrule
	\medskip
	\KwData{Sequential cost functions $f_1,\ldots, f_{\nt}$ and states $\xi_1,\ldots,\xi_{\nt}$}
 	\KwResult{Total cost $\sum_{t=1}^{\nt} f_t(x_t)$}
 		\For{$t\in [\nt]$}{
 				
		{Generate flexibility feedback:}
		\begin{align*}
		  \overline{p}_{t} = \psi^{\mathrm{SAC}}_{t}(x_{<t};\xi_{t})
		\end{align*}
 		
		Generate control signal and compute cost: 
		\begin{align*}
		 x_{t} =& \pi_{t}^{\mathsf{RHC}}(\overline{p}_{t})\\
		 \mathsf{cost} =& \mathsf{cost} + f_{t}(x_{t})
		\end{align*}
		
		{Update state:}
		
		{\quad $\xi_{t+1}(x_t;\xi_{t}) \longleftarrow \xi_{t}$}

		}\Return{$\mathsf{cost}$}
	\medskip
	\hrule
	\hrule
	\medskip
	\caption{The RHC scheme for online cost minimization.}
		\label{alg:online}
\end{algorithm}

We focus on a specific class of constraints and assume the operational constraints $h_i\left(\zeta,x\right) \leq  0, \ i=1,\ldots, k$ can be decoupled as ($t\in [\nt]$):
\begin{align}
\label{eq:decomposed_constraints}
    h^{(t)}_{i}\left(\zeta,x_t\right)\leq 0, \ i=1,\ldots,k_t.
\end{align}
First, we consider the following equivalent offline optimization of~\eqref{eq:offline}. 
Recall $q^*$ from~\eqref{eq:uniform}. 
\begin{align}
\label{eq:objective}
\min \sum_{t=1}^{\nt} &f_t(x_t) - \beta\log q^*(x)\\
\nonumber
\text{subject to } & h_i\left(\zeta,x\right) \leq  0, \ i=1,\ldots, k,
\end{align}
where $\beta>0$ is a tuning parameter. Decomposing the joint distribution $q^*(x)=\prod_{t=1}^{\nt}p^*_t(x_t|x_{<t})$ by the optimal conditional distributions given by~\eqref{eq:af}, the objective function~\eqref{eq:objective} becomes
\begin{align}
\nonumber
&\sum_{t=1}^{\nt} f_t(x_t) - \beta\log \left(\prod_{t=1}^{\nt} p_t^*(x_t|x_{<t})\right)\\
\label{eq:decomposition}
=&\sum_{t=1}^{\nt} \left(f_t(x_t) - \beta\log p_t^*(x_t|x_{<t})\right).
\end{align}
 
Eq.~\eqref{eq:decomposition} motivates the following RHC-based operator function at time $t\in [\nt]$, which includes the flexibility feedback $p_t(\cdot|x_{<t})$ as a penalty term in a greedy minimization: 
\begin{align}
\label{eq:RHC}
\pi_t^{\mathsf{RHC}}:=&\argmin_{x\in\mathbb{X}} f_t(x)-\beta\log {p}_t(x|x_{<t}) \\
&\text{subject to }~\eqref{eq:decomposed_constraints}.
\end{align}
Crucially, the following shows that feasibility is guaranteed when the flexibility feedback is optimal.
\begin{corollary}
\label{coro:feasibility}
Suppose at each time $t\in [\nt]$, the optimal flexibility feedback $p^*_t(\cdot|x_{<t})$ is sent to an operator constrained by~\eqref{eq:decomposed_constraints}. Then, the trajectory $x=(x_1,\ldots,x_{\nt})$ generated by the RHC-based operator function $\pi_t^{\mathsf{RHC}}$ in~\eqref{eq:RHC} is always feasible, \textit{i.e.,} $x\in\mathcal{S}(\phi,\xi)$.
\end{corollary}

\begin{proof}
Applying Theorem~\ref{coro:property}, it suffices to show that the signal $x_t = \pi_t^{\mathsf{RHC}}(p^*_{t}(\cdot|x_{<t}))$ generated at time $t$  satisfies
$
p^*_t(x_t|x_{<t})>0.
$
Suppose not, then there is a control signal $x^*=\pi_t^{\mathsf{RHC}}(p^*_{t}(\cdot|x_{<t}))$ such that $p^*_t(x|x_{<t})=0$ for some $t$ implies the objective in~\eqref{eq:RHC} becomes positive infinity. Our assumption $\mathcal{S}(\phi,\xi)\neq\emptyset$ implies that $p^*_t(\cdot|x_{<t})$ is not an all-zero vector. Therefore, $x$ is not the optimal solution of~\eqref{eq:RHC}, yielding a contradiction.
\end{proof}

\subsection{Aggregator: Data-driven approximation of the optimal flexibility feedback}
\label{sec:RL}

As we have already noted, computing the optimal flexibility feedback is computationally intensive. Thus, instead of computing it precisely, it is desirable to approximate it.  For the case of online cost minimization, it is possible to take a data-driven approach.  In particular, we propose the use of reinforcement learning to learn a function ${\psi}_t:{\Omega}\rightarrow\mathcal{P}$ that outputs the estimated flexibility feedback $\overline{p}_t$ given the current system state ${\xi}_{t}$. Note that we do not directly learn the disaggregation of $x_t$, which would have too large an action space. Instead, we fix a specific scheduling algorithm and learn the feedback vectors directly. This is another benefit of the concise representation of the feedback vectors.  

More specifically, we train an agent function ${\psi}_t$ using soft actor-critic (SAC)~\cite{haarnoja2018soft}, with the following generic reward function $r:\Omega^{t}\times \mathbb{X}^{t}\times\mathcal{P}\rightarrow \mathbb{R}$:
\begin{align}
\label{eq:reward}
r(\xi_{\leq t},\mathbf{x}_{\leq t},\mathbf{p}_t) = &\mathbb{H}(\mathbf{p}_t) - \sum_{i=1}^{m} c_i \min_{\mathbf{x}_{>t}} \left[g_i(\mathbf{x}; \phi,\xi_{\leq t})\right]_{+}.
\end{align}


The first term maximizes the entropy of the flexibility feedback vector, as a heuristic for the objective in~\eqref{eq:af1}. The second term penalizes the choice of $x_t$ that leads to an infeasible trajectory. Note that the reward function is independent of the price functions.  We provide more details in Appendix~\ref{app:learning}. 
We next demonstrate in simulations that feeding back to the operator the approximate optimal flexibility obtained from reinforcement learning is sufficient for achieving the desirable properties proven in Section \ref{sec:feedback}. 


\subsection{Experiments}

In the following, we show our experimental results for online EV charging, using real EV charging data
ACN-Data~\cite{lee2019acn}, which is a dataset collected from adaptive EV charging networks (ACNs) at Caltech and JPL. The detailed choices of SAC parameters and the design of the reward function for the SAC approach are presented in Appendix~\ref{app:learning}.


\begin{figure}[h]
    \centering
    \includegraphics[scale=0.6]{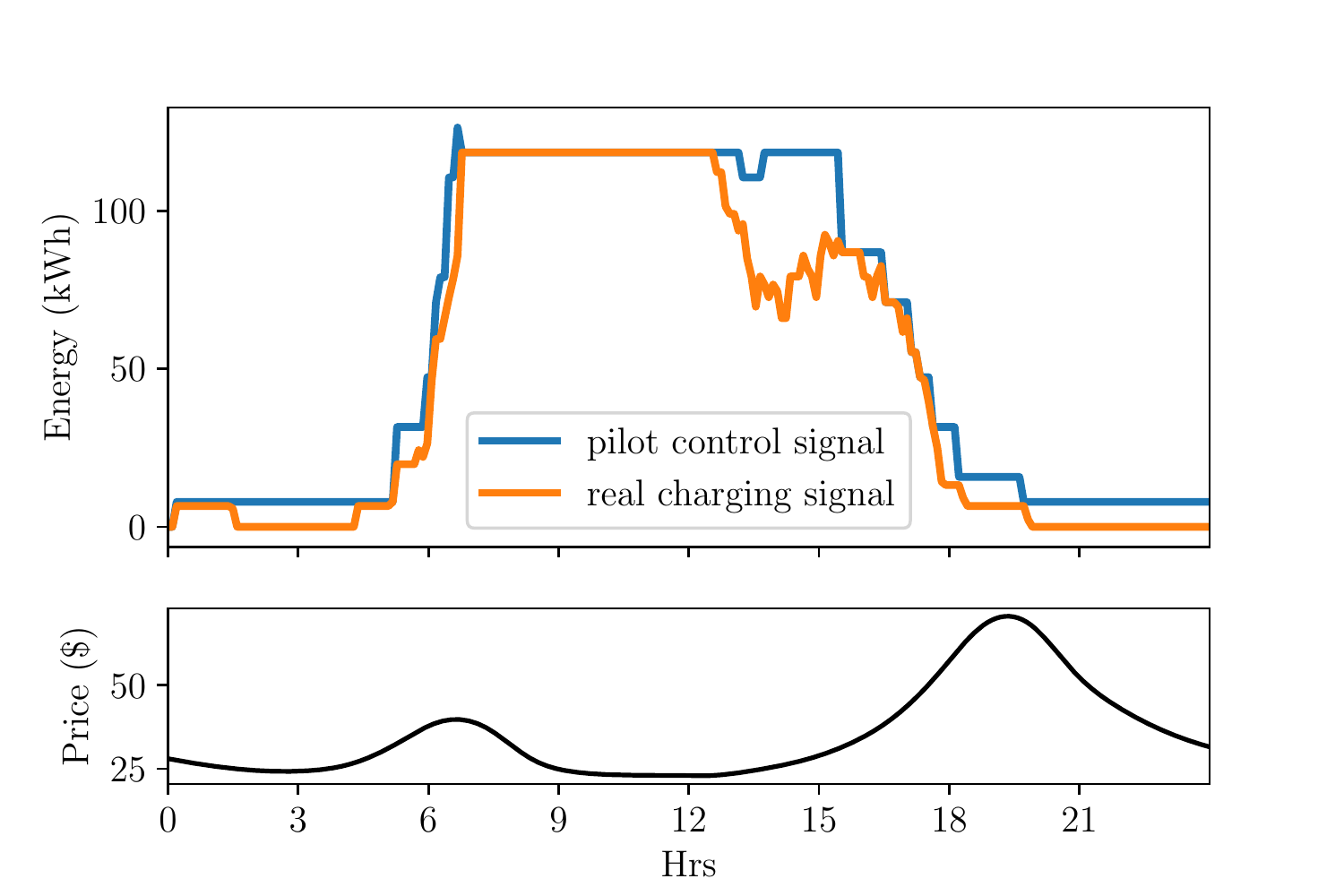}
    \caption{Pilot control signals and real energy allocated to EVs generated by Algorithm~\ref{alg:online}. }
    \label{fig:charging}
\end{figure}



{\bf Charging curves.} In Figure~\ref{fig:charging}, pilot control and real energy  signals are shown. The agent is trained on data collected at Caltech from Nov. 1, 2018 to Dec. 1, 2019 with linear price functions $f_t=1-t/24$, where $t\in [0,24]$ (unit: Hrs) is the time index and tested on Dec. 18, 2019 for JPL with average LMPs on the CAISO (California Independent System Operator) day-ahead market in 2016, shown on the bottom. The scheduling policy is fixed to be LLF (see Appendix~\ref{sec:monotonicity} for more details).  The set of power levels $\mathbb{X}$ is a discrete set that contains $60$ distinct power levels from $0$ kWh to $360$ kWh. We use tuning parameter $\beta=4000$. The pilot control signals are optimal solutions of~\eqref{eq:RHC}, which are always bounded from below by the real charging signals, representing the aggregate charging rates $\sum_i \phi_i(t)$ for $t\in [0,24]$. The figure highlights that, with a suitable choice of tuning parameter, the operator is able to schedule charging at time slots where prices are lower and avoid charging at the peak of prices, as desired. Note that the operational constraints used in this experiments is $x_t\leq 150$ (kWh) for every $t\in [\nt]$ and the learned flexibility feedback is able to automatically flatten the charging curve within this range, without explicitly knowing it.

\section{System capacity ${\digamma}(\phi,\xi)$ estimation}
\label{sec:real_time_eval}

In addition to minimizing cost, another important goal of the operator is to quantify the amount of flexibility available at each time. This is crucial for purposes of ensuring the ability to respond to failures and planning for capacity investment.  However, given that the private constraints of loads are not visible to the operator, such estimation is challenging.  Further, measuring the exact size of $\mathcal{S}(\phi,\xi)$ is intractable even if such constraints were visible, since the subset in $\mathbb{R}^{\nt}$ specified by inequalities~\eqref{eq:2.2} can be non-convex and even computing the volume of a convex body can be a hard problem~\cite{simonovits2003compute}. Furthermore, since a system's states are time-variant, the flexibility of the system also changes over time.

In this section, we illustrate how the optimal flexibility feedback can be used to estimate the system capacity ${\digamma}(\phi,\xi)$.
To this end, we propose an empirical estimation of the system capacity that uses an approximation of the optimal flexibility feedback and demonstrate 
our method using a case study of EV charging.

\subsection{Operator: Monte Carlo estimation}
\label{sec:estimation}

The task of the operator is to, given the optimal flexibility feedback, estimate the system capacity ${\digamma}(\phi,\xi)$ while also generating signals $x_1,\ldots,x_{\nt}$ that are always feasible with respect to both the operational
constraints \eqref{eq:offline_3} and the private aggregator constraints 
\eqref{eq:offline_2}.  
The approach we propose is an empirical estimation of the system capacity using Monte Carlo estimation. In particular, we consider
\begin{align}
\label{eq:em}
\mu_{N}(\phi,\xi):=
\frac{1}{N}\sum_{\ell=1}^{N}\sum_{t=1}^{\nt}\mathbb{H}\left(p_t(\cdot|x_{<t}(\ell))\right),
\end{align}
where the summation is over $\nt$ discrete time slots and $N$ trajectories.  For each, the corresponding entropy function computes the entropy of the flexibility feedback vector $p_t$ conditioned on the generated signals $x_{<t}(\ell)$) at each time $t\in [\nt]$:
\begin{align*}
\mathbb{H}\left(p_t(\cdot|x_{<t}(\ell))\right):=  -\sum_{x\in\mathbb{X}}p_t\left(x|x_{<t}(\ell)\right)\log{p_t\left(x|x_{<t}(\ell)\right)}.
\end{align*}

The goal of this approach is that, with suitable choices of operator functions
$\pi_t$, when the number $N$ of sampled trajectories becomes large, the approximation converges to the system capacity ${\digamma}(\phi,\xi)$. To see why, suppose at each time $t\in [\nt]$, the operation $\pi_t^{\mathsf{OPT}}$ is a stochastic function that samples a signal $X_t$ according to the optimal flexibility feedback $p^*_t$, i.e., for all $t\in [\nt]$ and $x_t\in \mathbb{X}$, $$\mathbb{P}\left(\pi_t^{\mathsf{OPT}}\left(p^*_t(\cdot|x_{<t}\right)=x_t\right)=p_t^*(x_t|x_{<t}).$$  In this context, the theorem below shows that we obtain an estimate of the system capacity ${\digamma}(\phi,\xi)$ using Monte Carlo estimation. 

\begin{theorem}
\label{thm:capacity_estimation}
If the $N$ trajectories $\{(x_1(\ell),\ldots,x_{\nt}(\ell))\}_{\ell=1}^{N}$ are generated i.i.d. by $\{\pi_1^{\mathsf{OPT}}, \ldots, \pi_{\nt}^{\mathsf{OPT}}\}$, then the empirical estimate in~\eqref{eq:em} converges to the system capacity almost surely, \textit{i.e.},
\begin{align*}
\mu_{N}(\phi,\xi)\xrightarrow{a.s.} {\digamma}(\phi,\xi) \text{ as } N\rightarrow\infty.
\end{align*}
\end{theorem}

Note that, in addition to providing a method for estimating the system capacity, the theorem also validates that the entropy of the flexibility feedback sent each time reflects the system's current flexibility. This indicates that, for instance, if the feedback vector is a uniform distribution on $\mathbb{X}$, then the system has maximal flexibility.

\begin{proof}[Proof of Theorem \ref{thm:capacity_estimation}]
Suppose $N$ trajectories $\{x(1),\ldots,x(N)\}$ are sampled i.i.d.\ according to the optimal flexibility feedback.
Equivalently, for all $\ell=1,\ldots, N$, the entropy of the optimal flexibility feedback $p^*_t(x_{<t}(\ell))$ can be written as the following conditional entropy
$
\mathbb{H}\left(p_t(\cdot|x_{<t}(\ell))\right) = \mathbb{H}\left(X_{t}|X_{<t}=x_{<t}\right),
$
where each $X_t\in \mathbb{X}$ is a random signal drawn according to $p^*_t(x_{<t}(\ell))$.
We claim that, if the random power signal $X_{t}$ is sampled according to $p_t^*(\cdot|x_{<t})$ conditioned on previous power signals $x_{<t}=(x_1,\ldots,x_{t-1})$ for all $t\in[\nt]$, then the accumulated flexibility over $t\in[\nt]$ is equal to the system capacity ${\digamma}(\phi,\xi)$ in expectation,
\begin{align}
\label{eq:real_flexibility}
\mathbb{E}_{Y}\left[\sum_{t=1}^{\nt}\mathbb{H}\left(X_{t}|X_{<t}=Y_{<t}\right)\right] = {\digamma}(\phi,\xi)
\end{align}
where the expectation is taken over the randomness of the signal trajectory $Y$ that has the same distribution as $X$.
The equality in~\eqref{eq:real_flexibility}  follows by noticing that the left hand side equals to the objective function in~\eqref{eq:af1}, with the flexibility feedback there at each time $t\in [\nt]$ being optimal. 
Noting that the expectation in~\eqref{eq:real_flexibility}  equals to ${\digamma}(\phi,\xi)$, the law of large numbers implies the theorem.
\end{proof}

\subsection{Aggregator: Look-ahead approximation of the optimal flexibility feedback}
\label{sec:look-ahead}

As we have discussed, computing the exact optimal flexibility feedback vectors $p_1^*,\ldots,p_{\nt}^*$ is computationally intensive and so approximations are desirable.  In Section~\ref{sec:RL} we have presented a data-driven approach for estimation via reinforcement learning.  Here, we take a different approach based on looking ahead rather than referring to historical data.  This approach is preferable in highly non-stationary situations. The approximation is presented in Appendix~\ref{app:approximation}. Notably, one may wonder if sending approximately optimal flexibility feedback to the operator is sufficient for achieving the desirable properties discussed in Section \ref{sec:feedback}.  In fact, it is and the results can be extended to hold for approximately optimal flexibility feedback computed as described above.  Perhaps the most important of these properties is feasibility, and so we provide a detailed discussion of the extension for feasibility in Appendix \ref{app:feasibility}.

\vspace{-0.2pt}

\subsection{Experiments}
\label{sec:flexibility_evaluation}

In our experiments, we apply Monte Carlo estimation and look-ahead approximation to the 
ACN-Data~\cite{lee2019acn}.  

\textbf{System capacity ${\digamma}(\phi,\xi)$ estimation.} Figure \ref{fig:year} shows the estimated $365$-day (average) system capacities $\mu_{N}(\phi,\xi)$ in (\ref{eq:em}) calculated by Monte Carlo estimation using the look-ahead approximation 
with $N=5$, $k=1$ and $\nt=240$
from Sep. 1, 2018 to Aug. 31, 2019. We use parameters that match the setup of the garage. The total number of charging stations is fixed as $54$, with peak power rate $6.6$ kWh. The set of power levels $\mathbb{X}$ is a discrete set that contains $60$ distinct power levels from $0$ kWh to $360$ kWh (for the definition of the parameters, see Appendix~\ref{app:approximation}). Note that, corresponding to this setting, in the case that every power trajectory in the length-$240$ time horizon is feasible, the maximal system capacity is $240\times\log 60 \approx 983$. 

An interesting observation from this figure is that, although there are fewer users after Nov. 1, 2018 (because of switching from free-charging to paid-charging), there is no significant decrease of system capacity. Additionally, notice that there is a decline of users during the holidays, and therefore total flexibility drops during the Christmas season. 

\begin{figure}[h]
	\centering
	\includegraphics[width=1\linewidth]{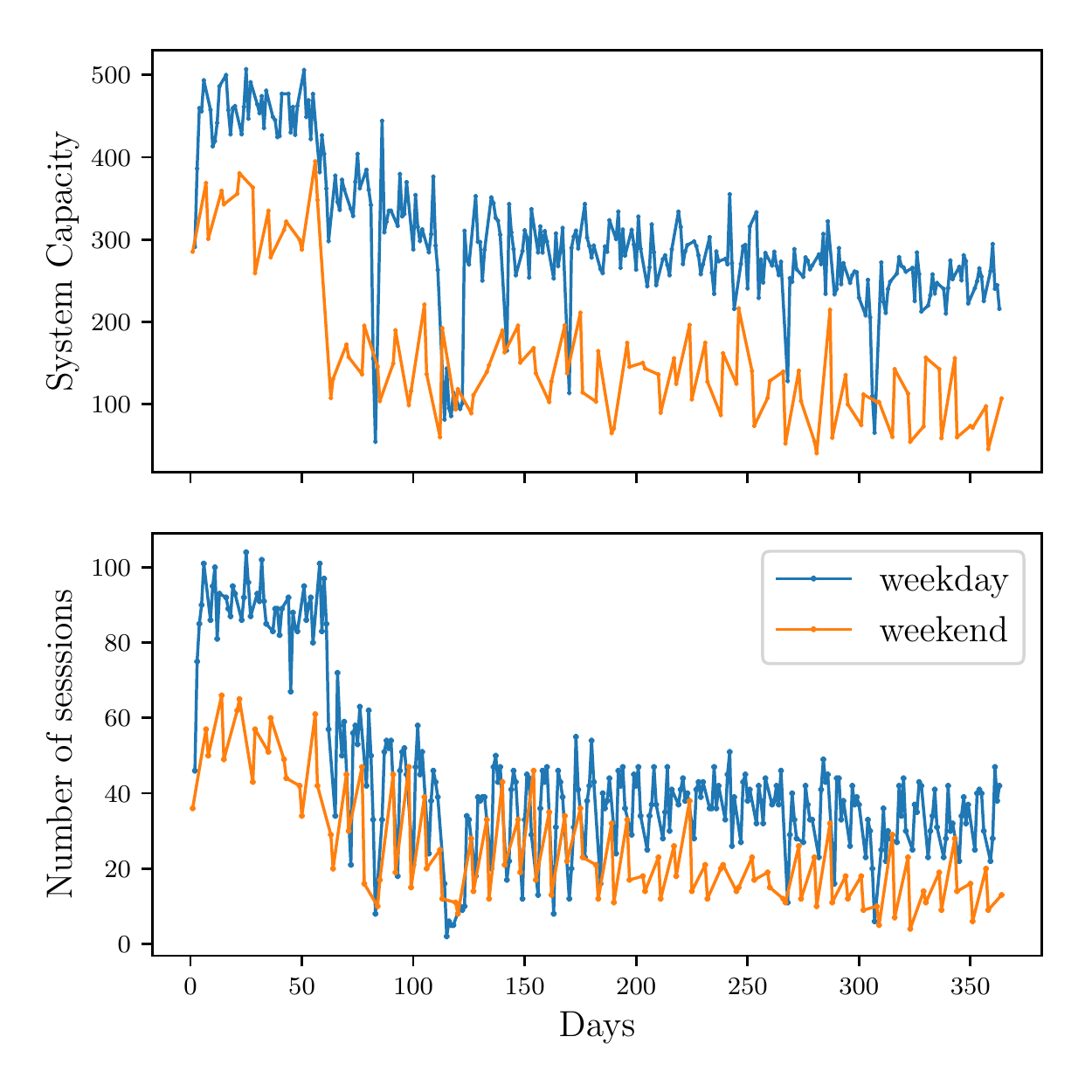}
	\caption{Estimated system capacities at Caltech with LLF compared with the number of charging sessions over a year, from Sep. 1, 2018 to Aug. 31, 2019. Weekends and weekdays are separated.}
	\label{fig:year}
	\medskip
\end{figure}

\textbf{Real-time flexibility feedback.}
Now, let us study the quality of the real-time flexibility feedback. Eq.~\eqref{eq:real_flexibility} gives the desired decomposition of system capacity, which enables us to characterize the spectrum of flexibility fluctuations. We show experimental results for real-time flexibility by considering a charging system within a single day. We use the same setting of parameters as described in Fig~\ref{fig:year}. 


\begin{figure}[t]
    \centering
    \includegraphics[scale=0.6]{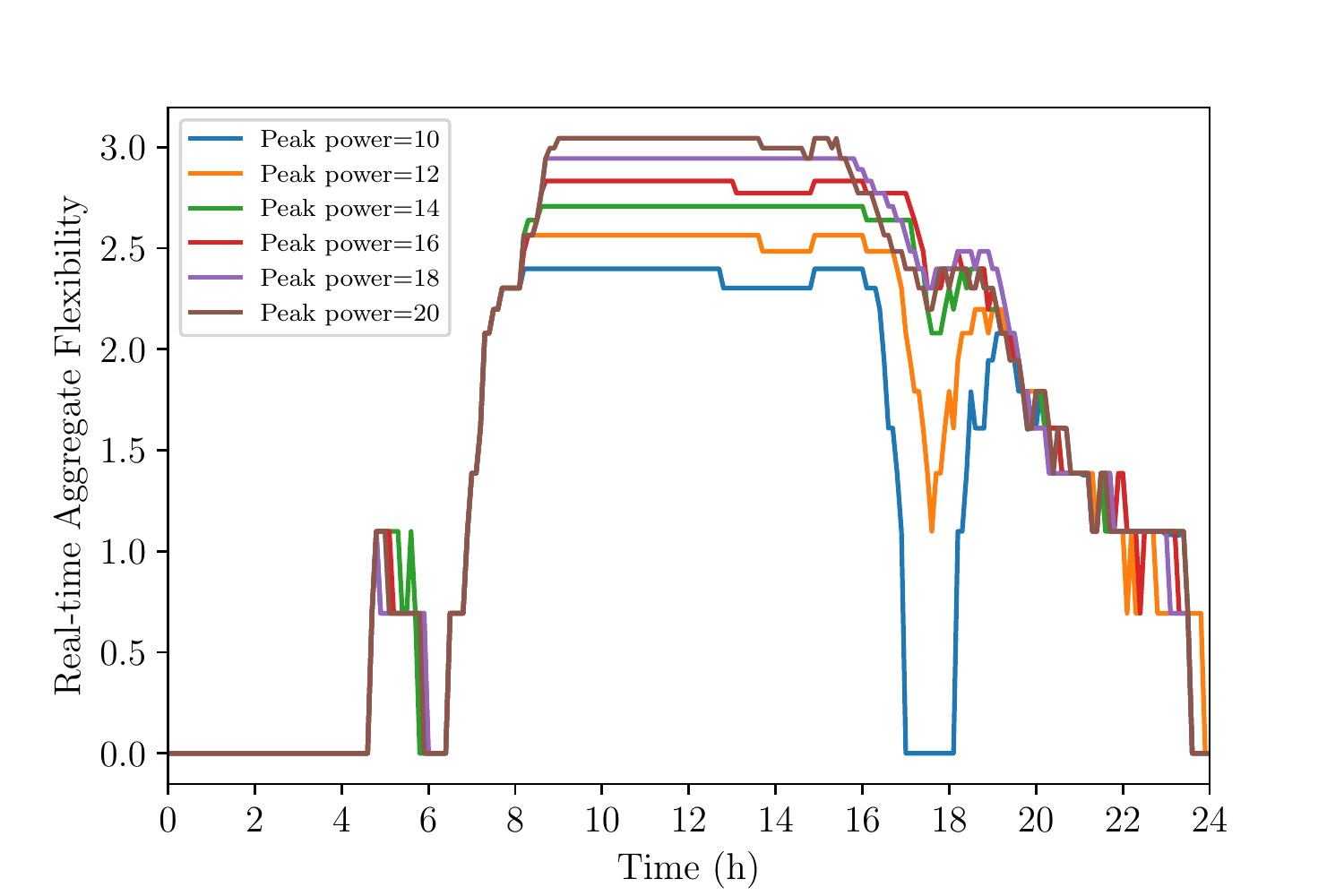}
    \caption{Impact of peak shaving constraints on real-time flexibility. We vary the peak power limit $\overline{\gamma}$(kWh) on Sep. 4, 2018.} 
    \label{fig:ope}
\end{figure}

We consider the case when operational constraints present, and the operator seeks to perform peak shaving. 
In Figure~\ref{fig:ope}, we vary the peak power limit defined in Example~\ref{example:operational_constraints} and it shows that the smaller the limit is set to be, the lower real-time aggregate flexibility the system has. Note that summing the real-time aggregate flexibility over time estimates the system capacity. Therefore a sharper limit induces a lower system capacity. Supplementary experimental results can be found in Appendix~\ref{app:experiments}.

\section{Concluding remarks}

This paper formalizes and studies the closed-loop control framework created by the interaction between a system operator and an aggregator.  Our focus is on the feedback signal provided by the aggregator to the operator that summarizes the real-time availability of flexibility among the loads controlled by the aggregator.  We present the design of an optimal flexibility feedback signal based on entropic maximization.  We prove a close connection between the optimal flexibility feedback signal and the system capacity, and show that when the signal is used the system operator can perform online cost minimization and system capacity estimation while provably respecting the private constraints of the loads controlled by the aggregator.  Further, we illustrate the effectiveness of these designs using simulation experiments of an EV charging facility.  

There is much left to explore about this optimal flexibility feedback signal presented in this work.  In particular, computing it is computationally intensive and we have presented two approaches for estimation.  Improving these and developing other approximations is of particular interest.  Further, exploring the use of flexibility feedback for operational objectives beyond cost minimization and capacity estimation is an important goal.  Finally, exploring the application of flexibility feedback in other settings, such as frequency regulation and real-time pricing, is exciting.

\begin{acks}
This work is supported by NSF through grants CCF 1637598, ECCS 1619352, ECCS 1931662, CPS ECCS 1739355, CPS ECCS 1932611.
\end{acks}
\newpage

\newpage

\bibliographystyle{ACM-Reference-Format}
\bibliography{final_lib}


\begin{thebibliography}{28}


\ifx \showCODEN    \undefined \def \showCODEN     #1{\unskip}     \fi
\ifx \showDOI      \undefined \def \showDOI       #1{#1}\fi
\ifx \showISBNx    \undefined \def \showISBNx     #1{\unskip}     \fi
\ifx \showISBNxiii \undefined \def \showISBNxiii  #1{\unskip}     \fi
\ifx \showISSN     \undefined \def \showISSN      #1{\unskip}     \fi
\ifx \showLCCN     \undefined \def \showLCCN      #1{\unskip}     \fi
\ifx \shownote     \undefined \def \shownote      #1{#1}          \fi
\ifx \showarticletitle \undefined \def \showarticletitle #1{#1}   \fi
\ifx \showURL      \undefined \def \showURL       {\relax}        \fi
\providecommand\bibfield[2]{#2}
\providecommand\bibinfo[2]{#2}
\providecommand\natexlab[1]{#1}
\providecommand\showeprint[2][]{arXiv:#2}

\bibitem[\protect\citeauthoryear{Acz{\'e}l, Forte, and Ng}{Acz{\'e}l
  et~al\mbox{.}}{1974}]%
        {aczel1974shannon}
\bibfield{author}{\bibinfo{person}{J{\'a}nos Acz{\'e}l}, \bibinfo{person}{Bruno
  Forte}, {and} \bibinfo{person}{Che~Tat Ng}.} \bibinfo{year}{1974}\natexlab{}.
\newblock \showarticletitle{Why the Shannon and Hartley entropies are
  ‘natural’}.
\newblock \bibinfo{journal}{\emph{Advances in applied probability}}
  \bibinfo{volume}{6}, \bibinfo{number}{1} (\bibinfo{year}{1974}),
  \bibinfo{pages}{131--146}.
\newblock


\bibitem[\protect\citeauthoryear{Bernstein, Le~Boudec, Paolone, Reyes-Chamorro,
  and Saab}{Bernstein et~al\mbox{.}}{2016}]%
        {bernstein2016aggregation}
\bibfield{author}{\bibinfo{person}{Andrey Bernstein},
  \bibinfo{person}{Jean-Yves Le~Boudec}, \bibinfo{person}{Mario Paolone},
  \bibinfo{person}{Lorenzo Reyes-Chamorro}, {and} \bibinfo{person}{Wajeb
  Saab}.} \bibinfo{year}{2016}\natexlab{}.
\newblock \showarticletitle{Aggregation of power capabilities of heterogeneous
  resources for real-time control of power grids}. In
  \bibinfo{booktitle}{\emph{2016 Power Systems Computation Conference (PSCC)}}.
  IEEE, \bibinfo{pages}{1--7}.
\newblock


\bibitem[\protect\citeauthoryear{Burger, Chaves-{\'A}vila, Batlle, and
  P{\'e}rez-Arriaga}{Burger et~al\mbox{.}}{2017}]%
        {burger2017review}
\bibfield{author}{\bibinfo{person}{Scott Burger}, \bibinfo{person}{Jose~Pablo
  Chaves-{\'A}vila}, \bibinfo{person}{Carlos Batlle}, {and}
  \bibinfo{person}{Ignacio~J P{\'e}rez-Arriaga}.}
  \bibinfo{year}{2017}\natexlab{}.
\newblock \showarticletitle{A review of the value of aggregators in electricity
  systems}.
\newblock \bibinfo{journal}{\emph{Renewable and Sustainable Energy Reviews}}
  \bibinfo{volume}{77} (\bibinfo{year}{2017}), \bibinfo{pages}{395--405}.
\newblock


\bibitem[\protect\citeauthoryear{Callaway and Hiskens}{Callaway and
  Hiskens}{2010}]%
        {callaway2010achieving}
\bibfield{author}{\bibinfo{person}{Duncan~S Callaway} {and}
  \bibinfo{person}{Ian~A Hiskens}.} \bibinfo{year}{2010}\natexlab{}.
\newblock \showarticletitle{Achieving controllability of electric loads}.
\newblock \bibinfo{journal}{\emph{Proc. IEEE}} \bibinfo{volume}{99},
  \bibinfo{number}{1} (\bibinfo{year}{2010}), \bibinfo{pages}{184--199}.
\newblock


\bibitem[\protect\citeauthoryear{Chen, Li, and Giannakis}{Chen
  et~al\mbox{.}}{2018b}]%
        {chen2018aggregating}
\bibfield{author}{\bibinfo{person}{Tianyi Chen}, \bibinfo{person}{Na Li}, {and}
  \bibinfo{person}{Georgios~B Giannakis}.} \bibinfo{year}{2018}\natexlab{b}.
\newblock \showarticletitle{Aggregating flexibility of heterogeneous energy
  resources in distribution networks}. In \bibinfo{booktitle}{\emph{2018 Annual
  American Control Conference (ACC)}}. IEEE, \bibinfo{pages}{4604--4609}.
\newblock


\bibitem[\protect\citeauthoryear{Chen, Dall'Anese, Zhao, and Li}{Chen
  et~al\mbox{.}}{2018a}]%
        {chen2018aggregate}
\bibfield{author}{\bibinfo{person}{Xin Chen}, \bibinfo{person}{Emiliano
  Dall'Anese}, \bibinfo{person}{Changhong Zhao}, {and} \bibinfo{person}{Na
  Li}.} \bibinfo{year}{2018}\natexlab{a}.
\newblock \showarticletitle{Aggregate Power Flexibility in Unbalanced
  Distribution Systems}.
\newblock \bibinfo{journal}{\emph{arXiv preprint arXiv:1812.05990}}
  (\bibinfo{year}{2018}).
\newblock


\bibitem[\protect\citeauthoryear{Csisz{\'a}r}{Csisz{\'a}r}{2008}]%
        {csiszar2008axiomatic}
\bibfield{author}{\bibinfo{person}{Imre Csisz{\'a}r}.}
  \bibinfo{year}{2008}\natexlab{}.
\newblock \showarticletitle{Axiomatic characterizations of information
  measures}.
\newblock \bibinfo{journal}{\emph{Entropy}} \bibinfo{volume}{10},
  \bibinfo{number}{3} (\bibinfo{year}{2008}), \bibinfo{pages}{261--273}.
\newblock


\bibitem[\protect\citeauthoryear{Evans, Tindemans, and Angeli}{Evans
  et~al\mbox{.}}{2019}]%
        {evans2019graphical}
\bibfield{author}{\bibinfo{person}{Michael~P Evans}, \bibinfo{person}{Simon~H
  Tindemans}, {and} \bibinfo{person}{David Angeli}.}
  \bibinfo{year}{2019}\natexlab{}.
\newblock \showarticletitle{A Graphical Measure of Aggregate Flexibility for
  Energy-Constrained Distributed Resources}.
\newblock \bibinfo{journal}{\emph{IEEE Transactions on Smart Grid}}
  (\bibinfo{year}{2019}).
\newblock


\bibitem[\protect\citeauthoryear{Gan, Topcu, and Low}{Gan
  et~al\mbox{.}}{2012}]%
        {gan2012optimal}
\bibfield{author}{\bibinfo{person}{Lingwen Gan}, \bibinfo{person}{Ufuk Topcu},
  {and} \bibinfo{person}{Steven~H Low}.} \bibinfo{year}{2012}\natexlab{}.
\newblock \showarticletitle{Optimal decentralized protocol for electric vehicle
  charging}.
\newblock \bibinfo{journal}{\emph{IEEE Transactions on Power Systems}}
  \bibinfo{volume}{28}, \bibinfo{number}{2} (\bibinfo{year}{2012}),
  \bibinfo{pages}{940--951}.
\newblock


\bibitem[\protect\citeauthoryear{Haarnoja, Zhou, Abbeel, and Levine}{Haarnoja
  et~al\mbox{.}}{2018}]%
        {haarnoja2018soft}
\bibfield{author}{\bibinfo{person}{Tuomas Haarnoja}, \bibinfo{person}{Aurick
  Zhou}, \bibinfo{person}{Pieter Abbeel}, {and} \bibinfo{person}{Sergey
  Levine}.} \bibinfo{year}{2018}\natexlab{}.
\newblock \showarticletitle{Soft actor-critic: Off-policy maximum entropy deep
  reinforcement learning with a stochastic actor}.
\newblock \bibinfo{journal}{\emph{arXiv preprint arXiv:1801.01290}}
  (\bibinfo{year}{2018}).
\newblock


\bibitem[\protect\citeauthoryear{Hao and Chen}{Hao and Chen}{2014}]%
        {hao2014characterizing}
\bibfield{author}{\bibinfo{person}{He Hao} {and} \bibinfo{person}{Wei Chen}.}
  \bibinfo{year}{2014}\natexlab{}.
\newblock \showarticletitle{Characterizing flexibility of an aggregation of
  deferrable loads}. In \bibinfo{booktitle}{\emph{53rd IEEE Conference on
  Decision and Control}}. IEEE, \bibinfo{pages}{4059--4064}.
\newblock


\bibitem[\protect\citeauthoryear{Hao, Lin, Kowli, Barooah, and Meyn}{Hao
  et~al\mbox{.}}{2014a}]%
        {hao2014ancillary}
\bibfield{author}{\bibinfo{person}{He Hao}, \bibinfo{person}{Yashen Lin},
  \bibinfo{person}{Anupama~S Kowli}, \bibinfo{person}{Prabir Barooah}, {and}
  \bibinfo{person}{Sean Meyn}.} \bibinfo{year}{2014}\natexlab{a}.
\newblock \showarticletitle{Ancillary service to the grid through control of
  fans in commercial building HVAC systems}.
\newblock \bibinfo{journal}{\emph{IEEE Transactions on smart grid}}
  \bibinfo{volume}{5}, \bibinfo{number}{4} (\bibinfo{year}{2014}),
  \bibinfo{pages}{2066--2074}.
\newblock


\bibitem[\protect\citeauthoryear{Hao, Sanandaji, Poolla, and Vincent}{Hao
  et~al\mbox{.}}{2014b}]%
        {hao2014aggregate}
\bibfield{author}{\bibinfo{person}{He Hao}, \bibinfo{person}{Borhan~M
  Sanandaji}, \bibinfo{person}{Kameshwar Poolla}, {and}
  \bibinfo{person}{Tyrone~L Vincent}.} \bibinfo{year}{2014}\natexlab{b}.
\newblock \showarticletitle{Aggregate flexibility of thermostatically
  controlled loads}.
\newblock \bibinfo{journal}{\emph{IEEE Transactions on Power Systems}}
  \bibinfo{volume}{30}, \bibinfo{number}{1} (\bibinfo{year}{2014}),
  \bibinfo{pages}{189--198}.
\newblock


\bibitem[\protect\citeauthoryear{Kara, Macdonald, Black, B{\'e}rges, Hug, and
  Kiliccote}{Kara et~al\mbox{.}}{2015}]%
        {kara2015estimating}
\bibfield{author}{\bibinfo{person}{Emre~C Kara}, \bibinfo{person}{Jason~S
  Macdonald}, \bibinfo{person}{Douglas Black}, \bibinfo{person}{Mario
  B{\'e}rges}, \bibinfo{person}{Gabriela Hug}, {and} \bibinfo{person}{Sila
  Kiliccote}.} \bibinfo{year}{2015}\natexlab{}.
\newblock \showarticletitle{Estimating the benefits of electric vehicle smart
  charging at non-residential locations: A data-driven approach}.
\newblock \bibinfo{journal}{\emph{Applied Energy}}  \bibinfo{volume}{155}
  (\bibinfo{year}{2015}), \bibinfo{pages}{515--525}.
\newblock


\bibitem[\protect\citeauthoryear{Kingma and Ba}{Kingma and Ba}{2014}]%
        {kingma2014adam}
\bibfield{author}{\bibinfo{person}{Diederik~P Kingma} {and}
  \bibinfo{person}{Jimmy Ba}.} \bibinfo{year}{2014}\natexlab{}.
\newblock \showarticletitle{Adam: A method for stochastic optimization}.
\newblock \bibinfo{journal}{\emph{arXiv preprint arXiv:1412.6980}}
  (\bibinfo{year}{2014}).
\newblock


\bibitem[\protect\citeauthoryear{Lee, Li, and Low}{Lee et~al\mbox{.}}{2019}]%
        {lee2019acn}
\bibfield{author}{\bibinfo{person}{Zachary~J Lee}, \bibinfo{person}{Tongxin
  Li}, {and} \bibinfo{person}{Steven~H Low}.} \bibinfo{year}{2019}\natexlab{}.
\newblock \showarticletitle{ACN-Data: Analysis and Applications of an Open EV
  Charging Dataset}. In \bibinfo{booktitle}{\emph{Proceedings of the Tenth ACM
  International Conference on Future Energy Systems}}. ACM,
  \bibinfo{pages}{139--149}.
\newblock


\bibitem[\protect\citeauthoryear{Madjidian, Roozbehani, and Dahleh}{Madjidian
  et~al\mbox{.}}{2018}]%
        {madjidian2018energy}
\bibfield{author}{\bibinfo{person}{Daria Madjidian}, \bibinfo{person}{Mardavij
  Roozbehani}, {and} \bibinfo{person}{Munther~A Dahleh}.}
  \bibinfo{year}{2018}\natexlab{}.
\newblock \showarticletitle{Energy storage from aggregate deferrable demand:
  Fundamental trade-offs and scheduling policies}.
\newblock \bibinfo{journal}{\emph{IEEE Transactions on Power Systems}}
  \bibinfo{volume}{33}, \bibinfo{number}{4} (\bibinfo{year}{2018}),
  \bibinfo{pages}{3573--3586}.
\newblock


\bibitem[\protect\citeauthoryear{Marzband, Sumper, Dom{\'\i}nguez-Garc{\'\i}a,
  and Gumara-Ferret}{Marzband et~al\mbox{.}}{2013}]%
        {marzband2013experimental}
\bibfield{author}{\bibinfo{person}{Mousa Marzband}, \bibinfo{person}{Andreas
  Sumper}, \bibinfo{person}{Jos{\'e}~Luis Dom{\'\i}nguez-Garc{\'\i}a}, {and}
  \bibinfo{person}{Ramon Gumara-Ferret}.} \bibinfo{year}{2013}\natexlab{}.
\newblock \showarticletitle{Experimental validation of a real time energy
  management system for microgrids in islanded mode using a local day-ahead
  electricity market and MINLP}.
\newblock \bibinfo{journal}{\emph{Energy Conversion and Management}}
  \bibinfo{volume}{76} (\bibinfo{year}{2013}), \bibinfo{pages}{314--322}.
\newblock


\bibitem[\protect\citeauthoryear{Meyn, Barooah, Bu{\v{s}}i{\'c}, Chen, and
  Ehren}{Meyn et~al\mbox{.}}{2015}]%
        {meyn2015ancillary}
\bibfield{author}{\bibinfo{person}{Sean~P Meyn}, \bibinfo{person}{Prabir
  Barooah}, \bibinfo{person}{Ana Bu{\v{s}}i{\'c}}, \bibinfo{person}{Yue Chen},
  {and} \bibinfo{person}{Jordan Ehren}.} \bibinfo{year}{2015}\natexlab{}.
\newblock \showarticletitle{Ancillary service to the grid using intelligent
  deferrable loads}.
\newblock \bibinfo{journal}{\emph{IEEE Trans. Automat. Control}}
  \bibinfo{volume}{60}, \bibinfo{number}{11} (\bibinfo{year}{2015}),
  \bibinfo{pages}{2847--2862}.
\newblock


\bibitem[\protect\citeauthoryear{Papadaskalopoulos, Strbac, Mancarella, Aunedi,
  and Stanojevic}{Papadaskalopoulos et~al\mbox{.}}{2013}]%
        {papadaskalopoulos2013decentralized}
\bibfield{author}{\bibinfo{person}{Dimitrios Papadaskalopoulos},
  \bibinfo{person}{Goran Strbac}, \bibinfo{person}{Pierluigi Mancarella},
  \bibinfo{person}{Marko Aunedi}, {and} \bibinfo{person}{Vladimir Stanojevic}.}
  \bibinfo{year}{2013}\natexlab{}.
\newblock \showarticletitle{Decentralized participation of flexible demand in
  electricity markets—Part II: Application with electric vehicles and heat
  pump systems}.
\newblock \bibinfo{journal}{\emph{IEEE Transactions on Power Systems}}
  \bibinfo{volume}{28}, \bibinfo{number}{4} (\bibinfo{year}{2013}),
  \bibinfo{pages}{3667--3674}.
\newblock


\bibitem[\protect\citeauthoryear{Sadeghianpourhamami, Refa, Strobbe, and
  Develder}{Sadeghianpourhamami et~al\mbox{.}}{2018}]%
        {sadeghianpourhamami2018quantitive}
\bibfield{author}{\bibinfo{person}{Nasrin Sadeghianpourhamami},
  \bibinfo{person}{Nazir Refa}, \bibinfo{person}{Matthias Strobbe}, {and}
  \bibinfo{person}{Chris Develder}.} \bibinfo{year}{2018}\natexlab{}.
\newblock \showarticletitle{Quantitive analysis of electric vehicle
  flexibility: A data-driven approach}.
\newblock \bibinfo{journal}{\emph{International Journal of Electrical Power \&
  Energy Systems}}  \bibinfo{volume}{95} (\bibinfo{year}{2018}),
  \bibinfo{pages}{451--462}.
\newblock


\bibitem[\protect\citeauthoryear{Sajjad, Chicco, and Napoli}{Sajjad
  et~al\mbox{.}}{2016}]%
        {sajjad2016definitions}
\bibfield{author}{\bibinfo{person}{Intisar~Ali Sajjad},
  \bibinfo{person}{Gianfranco Chicco}, {and} \bibinfo{person}{Roberto Napoli}.}
  \bibinfo{year}{2016}\natexlab{}.
\newblock \showarticletitle{Definitions of demand flexibility for aggregate
  residential loads}.
\newblock \bibinfo{journal}{\emph{IEEE Transactions on Smart Grid}}
  \bibinfo{volume}{7}, \bibinfo{number}{6} (\bibinfo{year}{2016}),
  \bibinfo{pages}{2633--2643}.
\newblock


\bibitem[\protect\citeauthoryear{Siano and Sarno}{Siano and Sarno}{2016}]%
        {siano2016assessing}
\bibfield{author}{\bibinfo{person}{Pierluigi Siano} {and}
  \bibinfo{person}{Debora Sarno}.} \bibinfo{year}{2016}\natexlab{}.
\newblock \showarticletitle{Assessing the benefits of residential demand
  response in a real time distribution energy market}.
\newblock \bibinfo{journal}{\emph{Applied Energy}}  \bibinfo{volume}{161}
  (\bibinfo{year}{2016}), \bibinfo{pages}{533--551}.
\newblock


\bibitem[\protect\citeauthoryear{Simonovits}{Simonovits}{2003}]%
        {simonovits2003compute}
\bibfield{author}{\bibinfo{person}{Mikl{\'o}s Simonovits}.}
  \bibinfo{year}{2003}\natexlab{}.
\newblock \showarticletitle{How to compute the volume in high dimension?}
\newblock \bibinfo{journal}{\emph{Mathematical programming}}
  \bibinfo{volume}{97}, \bibinfo{number}{1-2} (\bibinfo{year}{2003}),
  \bibinfo{pages}{337--374}.
\newblock


\bibitem[\protect\citeauthoryear{Subramanian, Garcia, Dominguez-Garcia,
  Callaway, Poolla, and Varaiya}{Subramanian et~al\mbox{.}}{2012}]%
        {subramanian2012real}
\bibfield{author}{\bibinfo{person}{Anand Subramanian}, \bibinfo{person}{Manuel
  Garcia}, \bibinfo{person}{A Dominguez-Garcia}, \bibinfo{person}{D Callaway},
  \bibinfo{person}{Kameshwar Poolla}, {and} \bibinfo{person}{Pravin Varaiya}.}
  \bibinfo{year}{2012}\natexlab{}.
\newblock \showarticletitle{Real-time scheduling of deferrable electric loads}.
  In \bibinfo{booktitle}{\emph{2012 American Control Conference (ACC)}}. IEEE,
  \bibinfo{pages}{3643--3650}.
\newblock


\bibitem[\protect\citeauthoryear{Subramanian, Garcia, Callaway, Poolla, and
  Varaiya}{Subramanian et~al\mbox{.}}{2013}]%
        {subramanian2013real}
\bibfield{author}{\bibinfo{person}{Anand Subramanian},
  \bibinfo{person}{Manuel~J Garcia}, \bibinfo{person}{Duncan~S Callaway},
  \bibinfo{person}{Kameshwar Poolla}, {and} \bibinfo{person}{Pravin Varaiya}.}
  \bibinfo{year}{2013}\natexlab{}.
\newblock \showarticletitle{Real-time scheduling of distributed resources}.
\newblock \bibinfo{journal}{\emph{IEEE Transactions on Smart Grid}}
  \bibinfo{volume}{4}, \bibinfo{number}{4} (\bibinfo{year}{2013}),
  \bibinfo{pages}{2122--2130}.
\newblock


\bibitem[\protect\citeauthoryear{Wenzel, Negrete-Pincetic, Olivares, MacDonald,
  and Callaway}{Wenzel et~al\mbox{.}}{2017}]%
        {wenzel2017real}
\bibfield{author}{\bibinfo{person}{George Wenzel}, \bibinfo{person}{Matias
  Negrete-Pincetic}, \bibinfo{person}{Daniel~E Olivares},
  \bibinfo{person}{Jason MacDonald}, {and} \bibinfo{person}{Duncan~S
  Callaway}.} \bibinfo{year}{2017}\natexlab{}.
\newblock \showarticletitle{Real-time charging strategies for an electric
  vehicle aggregator to provide ancillary services}.
\newblock \bibinfo{journal}{\emph{IEEE Transactions on Smart Grid}}
  \bibinfo{volume}{9}, \bibinfo{number}{5} (\bibinfo{year}{2017}),
  \bibinfo{pages}{5141--5151}.
\newblock


\bibitem[\protect\citeauthoryear{Zhao, Zhang, Hao, and Kalsi}{Zhao
  et~al\mbox{.}}{2017}]%
        {zhao2017geometric}
\bibfield{author}{\bibinfo{person}{Lin Zhao}, \bibinfo{person}{Wei Zhang},
  \bibinfo{person}{He Hao}, {and} \bibinfo{person}{Karanjit Kalsi}.}
  \bibinfo{year}{2017}\natexlab{}.
\newblock \showarticletitle{A geometric approach to aggregate flexibility
  modeling of thermostatically controlled loads}.
\newblock \bibinfo{journal}{\emph{IEEE Transactions on Power Systems}}
  \bibinfo{volume}{32}, \bibinfo{number}{6} (\bibinfo{year}{2017}),
  \bibinfo{pages}{4721--4731}.
\newblock


\end{thebibliography}

\newpage

\appendix


\section{Learning flexibility feedback}
\label{app:learning}

Soft actor-critic (SAC)~\cite{haarnoja2018soft} is an off-policy maximum entropy deep reinforcement learning algorithm, which in many complicated learning scenarios (such as control of humanoid robotics) outperforms deep deterministic policy gradient (DDPG) approaches, especially when the action space is a continuous and high-dimensional. The policy in our experiments is fixed to be a parameterized family of Gaussian distributions.

\subsection{Approximate agent}

We train an agent $\psi^{\mathrm{SAC}}_t:\overline{\Omega}\rightarrow\mathcal{P}$ using SAC whose input at time $t\in [\nt]$ is a state parameter $\xi_t$ that encodes the remaining energy to be delivered and the remaining charging time for the EV being charged at each station $i\in [\ns]$ and time $t\in [\nt]$, denoted by $\psi^{\mathrm{SAC}}_t(i)$. Knowing the states $\xi_{<t}$, scheduling algorithm $\phi$ and the signals $x_{<t}$ gives the state $\xi_t$. 

\subsection{Parameters in our experiments}

In the experiments, the state space is $\mathbb{R}_{+}^{2\times{\ns}}$ where ${\ns}$ is the total number of charging stations and a state vector for each charging station is $(e_t,[d(j)-t]^{+})$, i.e., the remaining energy to be charged and the remaining charging time if it is being used; otherwise the vector is an all-zero vector. The action space is $\mathbb{R}_{+}^{|\mathbb{X}|}$. Moreover, the outputs of the neural networks are normalized into the probability simplex $\mathcal{P}$ afterwards. Hyper-parameters in our experiments are shown in Table~\ref{table:parameter}.
\begin{table}[h]
    \centering
    \begin{tabular}{l|l}
    \hline
    \hline
    \multicolumn{2}{c}{Soft actor-critic} \\
        \hline
      Parameter & Value\\
      optimizer & Adam~\cite{kingma2014adam}\\
      learning rate & $3\cdot 10^{-4}$\\
      discount ($\gamma$) & $0.5$\\
      relay buffer size & $10^6$\\
      number of hidden layers & $2$\\
      number of hidden units per layer & $256$\\
      number of samples per minibatch & $256$\\
      non-linearity & ReLU\\
      temperature parameter ($\alpha$) & $0.5$\\
      \hline
     \multicolumn{2}{c}{Markov decision process} \\
     \hline
     power levels ($\mathbb{X}$) & $\{1,2,\ldots,20\}$  (in kWh) \\
     number of stations (${\ns}$)  & Caltech (54) / JPL (52)\\
     state space & $\mathbb{R}_{+}^{2\times {\ns}}$\\
     action space & $[0,1]^{|\mathbb{X}|}$\\
     reward $r_{\mathrm{EV}}(\xi_t,p_t)$ & $\sigma_1=0.1$, $\sigma_2=0.2$, $\sigma_3=2$\\
     time interval ($\Delta$) & $6$ minutes\\
     operational constraints & $x_t\leq 150$ (kWh), \ $\forall \ t\in [\nt]$\\
      \hline
      \hline
    \end{tabular}
   \caption{Hyper-parameters in the experiments.}
  \label{table:parameter}
\end{table}

\subsection{Reward function in training}

For the deferrable loads in Example~\ref{example:ev}, once constraints are violated, they can on longer be satisfied by future decisions.  Therefore, the minimization can be removed and we have the following specific reward function for EV charging scenario:
\begin{align}
\nonumber
r_{\mathrm{EV}}(\xi_{\leq t},p_t) = &\mathbb{H}(p_t) + \sigma_1 \left|\left|\phi_t(\xi_{\leq t},x_{\leq t})\right|\right|_2\\
\label{eq:reward_EV}
-\sigma_2& \Big[e(j) - \sum_{t=1}^{\nt}\phi_t(j) \Big]_{+}-\sigma_3\Big|\pi_t^{\mathsf{RHC}}(p_t) - \sum_{j=1}^{\nn}\phi_t(j)\Big|
\end{align}
where $\sigma_0, \sigma_1,\sigma_2$ and $\sigma_3$ are positive constants. The second term is to enhance charging performance and the last two terms are realizations of the last term in~\eqref{eq:reward} for constraints~\eqref{eq:f3} and~\eqref{eq:f2}. The other constraints in Example~\ref{example:ev} can automatically be satisfied by enforcing the constraints in the fixed scheduling algorithm $\phi$.

With the settings described above, in Figure~\ref{fig:reward} we show a typical training curve of the reward function in~\eqref{eq:reward_EV}. The constants in~\eqref{eq:reward_EV} are $\sigma_1=0.1$, $\sigma_2=0.2$ and $\sigma_3=2$.


\section{Look-ahead approximation}
\label{app:approximation}
Before presenting the design, we first introduce some generalized notation for feasible power levels that, that extends ~\eqref{eq:set_of_feasible} to the case of $k$-step look-ahead. 
\begin{align*}
\mathcal{S}_k(\phi,\xi|x_{<t}):=
\Big\{x_{t\rightarrow t+k-1}&\in\mathbb{X}^{k}: \exists x_{t+k\rightarrow\nt} \text{ s.t. } \\
& g_i\left(\phi,\xi,x\right) \leq 0,
\forall i=1,\ldots, m\Big\}.
\end{align*}
Further, the closed-form expression in Theorem~\ref{thm:con} motivates us to consider the following approximation of the optimal flexibility feedback for all $x\in\mathbb{X}$ and $x_{< t}\in\mathbb{X}^{t-1}$:
\begin{align}
\label{eq:approx}
\widehat{p}_{t,k}(x|x_{<t}) = \frac{\left|\mathcal{S}_k(\phi,\xi|(x_{< t}, x))\right|}{\left|\mathcal{S}_k(\phi,\xi|x_{<t})\right|}.
\end{align}
This, in turn, leads to the following recursive formula for all $2 \leq k\leq \nt-t+1$:
\begin{align}
\label{eq:recur}
\left|\mathcal{S}_k(\phi,\xi|x_{<t})\right|=
 \sum_{x\in \mathcal{S}_1(\phi,\xi|x_{<t})}&\left|\mathcal{S}_{k-1}(\phi,\xi|(x_{<t}, x))\right|.
\end{align}
Here, we use $k$ as the \textit{look-ahead depth} and note that, when $k=\nt$, the approximation in~(\ref{eq:approx}) becomes exact. 

Now, using this notation, in order to estimate the system capacity ${\digamma}(\phi,\xi)$, we need to estimate the size of $\mathcal{S}_k(\phi,\xi|x_{<t})$ and $\mathcal{S}_k(\phi,\xi|x_{\leq t})$ using the recursive formula in~\eqref{eq:recur}. Accomplishing this depends on calculating the feasible set for selecting $x_{t}$, given fixed $x_{<t}$, i.e., characterizing the set $\mathcal{S}_{k}(\phi,\xi|x_{<t})$ with look-ahead depth $k=1$.

To provide a characterization of the first-order approximation $\mathcal{S}_1(\phi,\xi|x_{<t})$, we make a monotonicity assumption on the disaggregation policy $\phi$, defined as follows.
\begin{definition}[Monotonicity]
A (causal) disaggregation policy $\phi$ is \textbf{monotone} if for any $t\in [\nt]$, $x_{t}\leq y_t$ implies that for all $x_{<t}\in\mathcal{S}_{t}(\phi,\xi)$,
\begin{align}
\label{eq:monotone}
\phi_{t}\left(\xi_{<t},\left(x_{<t}, x\right)\right) \preceq \phi_{t}\left(\xi_{<t},\left(x_{<t}, y\right)\right).
\end{align}
\end{definition}

Assuming that a scheduling algorithm $\phi$ is monotone, the feasible set of $x_{t}$ conditioning on $x_{<t}$ can be characterized by a closed interval, as stated in the following theorem.

\begin{theorem}
\label{thm:interval}
Consider a system of deferrable loads with constraints specified by Example~\ref{example:ev} and $\xi$ being the associated states. For $t\in[\nt]$, for any $\xi$ and monotone scheduling algorithm $\phi$, the set $\mathcal{S}_1(\phi,\xi|x_{<t})$ can be written as the intersection of the set of power signals $\mathbb{X}$ and a closed interval:
$
\mathcal{S}_1(\phi,\xi|x_{<t}) = \mathbb{X}\bigcup\mathcal{I}_t(\phi,\xi),
$
where $\mathcal{I}_t(\phi,\xi) := \left[\alpha_t,\beta_t\right]$ is a closed interval in $\mathbb{R}_{\geq 0}$.
\end{theorem}

\begin{proof}[Proof of Theorem~\ref{thm:interval}]
To prove the theorem, it is equivalent to show that if $x\in\mathbb{X}$ and $y\in\mathbb{X}$ with $x<y$ are two feasible power levels in $\mathcal{S}_1(\phi,\xi|x_{<t})$, then any $z\in\mathbb{X}$ with $x\leq z\leq y$ is also in $\mathcal{S}_1(\phi,\xi|x_{<t})$. Since the disaggregation policy $\phi$ is monotonically causal, the inequality $z\leq y$ guarantees that the power scheduled with $x_{t}=z$ to each load is always larger or equal to the case when $x_{t}=x$. Therefore, considering that $x_{t}=y$ is a feasible choice, since $z$ does not violate any constraint for satisfying the demands of deferrable loads, it must also be a feasible power signal.
\end{proof}

It is typically straightforward to verify that the classical scheduling policies such as the least-laxity-first (LLF) scheduling and the earliest-deadline-first (EDF) scheduling are monotone, and thus the theory above applies. We demonstrate this for two classical policies, LLF and EDF, and one new policy termed feasibility interval maximization (FIM) in Appendix \ref{sec:monotonicity}.  FIM is a new policy motivated by the feasibility analysis in this paper. Provided with a power signal $x_{t}>\alpha_t$, FIM assigns power to the loads with negative laxity, proportionally to $-\rho_{\mathrm{Lax}}(j,t)$. Note that our purpose in discussing FIM is to demonstrate a contrast with LLF and EDF in our experimental results, not to present an ``optimal'' policy. 


\section{Feasibility via approximate flexibility feedback}
\label{app:feasibility}
Although the approximation of flexibility feedback in~\eqref{eq:approx} is not precise, in this section we show that it is accurate enough to ensure feasibility under certain conditions. Specifically, consider a system of deferrable loads with constraints specified by Example~\ref{example:ev}. The following lemma states if none of the loads demands ``excessive'' energy upon arrival and the system has enough capacity for charging every load at their peak rates, then the system is always feasible by choosing the power signal according to the approximate flexibility feedback. 

\begin{lemma}
Consider a system of deferrable loads with constraints specified by Example~\ref{example:ev}.
Suppose the following conditions hold:
\begin{enumerate}
\item There is no operational constraints and there exists $x\in\mathbb{X}$ such that $x\geq \sum_{j=1}^{\nn}r(j)$.
\item At each time  $t\in [\nt]$, the selected $x_{t}\in\mathbb{X}$ satisfies (with look-ahead depth $k\geq 1$) $$\widehat{p}_{t,k}(x|x_{<t})>0.$$
    \item The aggregator state $\xi$  satisfies that $$\rho_{\mathrm{Lax}}(j,t)\geq 0, \text{ for } t\in [a(j), a(j)+1)$$ and $j\in [\nn]$ where $\rho_{\mathrm{Lax}}(j,t)$ is the laxity of the load $j$ defined in~\eqref{eq:laxity}.
\end{enumerate}
It is guaranteed that $\mathcal{S}_1(\phi,\xi|x_{<t})\neq\emptyset$ for all $t\in [\nt]$ and any monotone disaggregation policy $\phi$.
\end{lemma}

\begin{proof}
Assuming that at time $t-1\in [\nt]$, the chosen power signal $x_{t-1}$ satisfies $\widehat{p}_{t-1,k}(x_{t-1}|x_{<t-1})>0$, it remains to validate that there always exists some $x_{t}\in\mathbb{X}$ such that the approximated flexibility feedback $\widehat{p}_{t,k}(x|x_{<t})>0$, conditioning on the previously selected power levels $x_1,\ldots,x_{t-1}$. To see this, note that condition (1) and (2) ensures that if the system is feasible at the previous time step $t-1$ (\textit{i.e., } $\mathcal{S}_1(\phi,\xi|x_{<t-1})\neq\emptyset$), then there is always a feasible power level in $\mathbb{X}$ for $x_{t}$, assuming there is no new loads arrive at the current time $t$. Condition (3) further guarantees that the demands of the new loads can also be satisfied, as long as the disaggregation policy $\phi$ is monotone. Therefore $\widehat{p}_{t,k}(x|x_{<t})>0$. By induction over $t\in [\nt]$ the proof is completed.
\end{proof}

\section{Monotonicity of common policies}
\label{sec:monotonicity}

In this section we show that LLF, EDF, and FIM are monotone policies.  Throughout, we fix the environment parameter $\xi$ for deferrable loads and denote by $(a(j),d(j),e_{t}(j),r(j))$ the charging states of the $j$-th load at time $t\in[\nt]$ where $a(j),d(j)$ and $r(j)$ are defined in Example~\ref{example:ev} and $e_t(j)$ is the remaining energy to be delivered at time $t\in [\nt]$. Additionally, let $\delta_{t}(j):=[d(j)-t]^{+}$ be the remaining charging duration (excluding the current time slot). 

\paragraph{Least-laxity-first (LLF) scheduling}
The laxity of the load $j\in [\nn]$ at time $t\in[\nt]$ is defined as
\begin{align}
\label{eq:laxity}
\rho_{\mathrm{Lax}}(j,t):=\begin{cases}
\delta_{t}(j)-{e_{t}(j)}/r(j), & \  t\geq a(j)  \\
+\infty, & \  t< a(j)
\end{cases}.
\end{align}
If the laxity is negative, the car will never be fully charged and the the system becomes infeasible. Therefore, for any monotone disaggregation policy $\phi$, the corresponding bounds must satisfy
\begin{subequations}
\begin{align}
\label{eq:6.1}
    \alpha_t &\geq \sum_{j\in\underline{\mathcal{N}}(t)}\min\big\{r(j),-\rho_{\mathrm{Lax}}(j,t){r(j)}\big\},\\
\label{eq:6.2}
    \beta_t &\leq  \sum_{j\in\underline{\mathcal{N}}(t)} \min\big\{r(j),e_{t}(j)\big\},
\end{align}
\end{subequations}
where $\underline{\mathcal{N}}(t):=\{j\in [\nn]:\rho_{\mathrm{Lax}}(j,t)\leq 0\}$. It is immediate to see that the equalities~\eqref{eq:6.1} and~\eqref{eq:6.2} hold for LLF. 

\paragraph{Earliest-deadline-first (EDF) scheduling}
Under EDF the summation in (\ref{eq:6.1}) for $\alpha_t$ needs to be replaced by a summation over $\underline{\mathcal{N}}(t)\bigcup\mathcal{N}_{\mathrm{EDF}}(t)$ where a load $j$ is in $\mathcal{N}_{\mathrm{EDF}}(t)$ if there exists $i\in\underline{\mathcal{N}}(t)$ such that $d_{j}(t)\leq d_{i}(t)$:
\begin{align*}
\mathcal{N}_{\mathrm{EDF}}(t):=\left\{j\in [\nn]: \exists i\in\underline{\mathcal{N}}(t) \text{ s.t. } d_{j}(t)\leq d_{i}(t) \right\}.
\end{align*}

\paragraph{Feasibility-interval-maximization (FIM) scheduling}

Recall that, when provided with a power signal $x_{t}>\alpha_t(\phi,\xi)$, FIM assigns power to the loads with negative laxity, proportionally to $-\rho_{\mathrm{Lax}}(j,t)$. To understand the motivation behind FIM, observe that increasing the laxity of the loads decreases the lower bound $\alpha_t(\phi,\xi)$. Therefore, intuitively, it is desirable to ensure that as many loads as possible have non-negative laxity. Clearly, FIM is monotone, since the higher $x_t$ is, the larger amount of energy is assigned to the EVs. The upper and lower bounds of the interval can be computed the same as in~\eqref{eq:6.1} and~\eqref{eq:6.2}.

\section{Supplementary Simulation results}
\label{app:experiments}

In Table~\ref{tab:year}, we summarize the quality of the capacity estimation, undelivered energy percentage, and tracking error  (see \eqref{eq:mpe} and \eqref{eq:mse} for definitions) for three scheduling policies, EDF, LLF and FIM, on both the Caltech and JPL garages. We measure performance using the mean squared error (MSE) as the tracking error:
\begin{align}
\label{eq:mse}
\mathsf{MSE}(\phi,x):= \sum_{k=1}^{N}\sum_{t=1}^{\nt}\Big|\sum_{j=1}^{\nn}\phi_{t}^{(k)}(j)-x^{(k)}_t\Big|^2 / (N\times \nt),
\end{align}
where $x^{(k)}_t$ is the $t$-th power signal for the $k$-th test and $\phi_{t}^{(k)}(j)$ is the energy scheduled to the $j$-th load at time $t$ for the $k$-th test. Additionally, define the mean percentage error with respect to the undelivered energy as
\begin{align}
\label{eq:mpe}
\mathsf{MPE}(\phi,\xi):= \sum_{k=1}^{N}\sum_{t=1}^{\nt}\sum_{j=1}^{\nn}\phi_{t}^{(k)}(j)\big/{ \Big((N\times \nt)\cdot\sum_{j=1}^{\nn}e_j\Big)},
\end{align}
where $e_j$ is the energy request for each load $j$.

The results are averaged over the days from Sep. 1, 2018 to Aug. 31, 2019. The results show that FIM achieves the highest (estimated) system capacity, and lowest tracking error. However, as FIM always maximizes the feasible charging interval, as a trade-off, its average percentage of undelivered energy is always the largest.

\begin{table}[h]
\centering
\small
\begin{tabular}{c c c c} 
\toprule
    {$\phi$}  & {System Capacity} & {Undelivered (\%)} & {Tracking Error (kWh)} \\ \midrule
    {EDF}  & {$262.6726$} & {$\mathbf{5.6044}$} &  {$11.1570$}\\
    {LLF}  & {$271.7024$}  & {$10.1406$} & {$3.1356$}  \\
    {FIM}  & {$\mathbf{271.9571}$}  & {$10.3148$} & {$\mathbf{2.7661}$}  \\ \midrule
    {EDF}  &  {$221.5389$} & {$\mathbf{6.6837}$} &   {$15.5949$} \\
    {LLF}  &  {$242.4101$} & {$11.8726$} & {$7.0418$}  \\
    {FIM}  &  {$\mathbf{242.8318}$}  & {$12.1583$} & {$\mathbf{6.6034}$} \\ \bottomrule
\end{tabular}
\caption{Estimated system capacity ${\digamma}(\phi,\xi)$, undelivered energy percentage $\mathsf{MPE}(\phi,\xi)$, and tracking error $\mathsf{MSE}(\phi,x)$ for Caltech (top) and JPL (bottom) comparing EDF, LLF, and FIM.}
\label{tab:year}
\end{table}





\end{document}